\documentclass[final,12pt]{colt2019} 

\title[Sorted Top-$k$ in Rounds]{Sorted Top-$k$ in Rounds}
\usepackage{times}

\coltauthor{%
\Name{Mark Braverman}  \Email{ mbraverm@cs.princeton.edu}\\
\addr Princeton University
\AND
 \Name{Jieming Mao} \Email{maojm517@gmail.com}\\
 \addr University of Pennsylvania
 \AND
 \Name{Yuval Peres} \Email{yperes@gmail.com}\\
}

\usepackage{amsmath,amssymb,amsfonts,latexsym,bbm,xspace,graphicx,float,mathtools}
\usepackage{color}
\usepackage{enumerate}
\usepackage{comment}
\usepackage{algorithmic,algorithm}

\usepackage{appendix}
\usepackage[nomessages]{fp}

\usepackage{bbm}

\usepackage{tikz}
\usetikzlibrary{patterns}
\DeclareMathOperator*{\argmin}{argmin}

\newtheorem{fact}{Fact}[section]

\newtheorem{openproblem}[theorem]{Open Problem}

\def\D{\mathbb{D}}
\def\E{\mathbb{E}}
\def\1{\mathbbm{1}}

\newlength\myindent
\begin{document}



\maketitle

\begin{abstract}
We consider the sorted top-$k$ problem whose goal is to recover the top-$k$ items with the correct order out of $n$ items using pairwise comparisons. In many applications, multiple rounds of interaction can be costly. We restrict our attention to algorithms with a constant number of rounds $r$ and try to minimize the sample complexity, i.e. the number of comparisons.

When the comparisons are noiseless, we characterize how the optimal sample complexity depends on the number of rounds (up to a polylogarithmic factor for general $r$ and up to a constant factor for $r=1$ or 2). In particular, the sample complexity is $\Theta(n^2)$ for $r=1$, $\Theta(n\sqrt{k} + n^{4/3})$ for $r=2$ and $\tilde{\Theta}\left(n^{2/r} k^{(r-1)/r} + n\right)$ for $r \geq 3$. 

We extend our results of sorted top-$k$ to the noisy case where each comparison is correct with probability $2/3$. When $r=1$ or 2, we show that the sample complexity gets an extra $\Theta(\log(k))$ factor when we transition from the noiseless case to the noisy case.

We also prove new results for top-$k$ and sorting in the noisy case. We believe our techniques can be generally useful for understanding the trade-off between round complexities and sample complexities of rank aggregation problems.
\end{abstract}


\newpage

\section{Introduction}

Rank aggregation is a fundamental problem which finds numerous applications in recommendation systems, web search, social choice, peer grading and crowdsourcing. The most studied problem in rank aggregation is sorting. It aims to find the total ordering of all items. People also consider the top-$k$ problem when it is only necessary to recover the set of top-$k$ items. However, for some applications, the rank aggregation task required is neither sorting nor top-$k$. For example, when a recommendation system shows the user a list of items, it might want to display these items in the order of recommendation. As another example, in a tournament, people usually care about the exact rankings of top players but not others.

These examples motivate us to study the sorted top-$k$ problem which lies between sorting and top-$k$. In this problem, we have $n$ items with an underlying order and the goal is to recover the top-$k$ items with the correct order using pairwise comparisons.  

In many applications, multiple rounds of interaction are costly. For example, if we collect comparison data via crowdsourcing, the comparisons can be done in parallel by different crowd workers and the total amount of time spent is mainly decided by the number of rounds. Therefore we consider the sorted top-$k$ problem in the parallel comparison model introduced by \cite{Valiant75}. In this model, an algorithm performs a set of comparisons in each round and the actual set can depend on comparison results of previous rounds. The goal is to solve the task in bounded number of rounds while minimizing the sample complexity, i.e. the total number of comparisons. 

Parallel comparison algorithms have been intensively studied around 30 years ago. Much of the attention has been put into two problems: sorting \cite{HaggkvistH81, AjtaiKS83,BollobasT83,Kruskal83,Leighton84,BollobasH85,Alon85,AlonAV86,AzarV87,Pippenger87,AlonA88b,AlonA88a,Akl90} and top-$k$ \cite{Valiant75,Reischuk81,AjtaiKSS86,Pippenger87,AlonA88b,AlonA88a,AzarP90,BollobasB90}. While being very related to sorting and top-$k$, sorted top-$k$ has not been studied in the parallel comparison model.

Without the constraint on the number of rounds, sorted top-$k$ can be easily solved by combining sorting and top-$k$ algorithms: we can first use a top-$k$ algorithm to find the set of top-$k$ items and then use a sorting algorithm to sort these $k$ items. If we are using sorting and top-$k$ algorithms with optimal sample complexities, one can easily show that the combination gives a sorted top-$k$ algorithm with optimal sample complexity (up to a constant factor). 

However, if we only have bounded number of rounds, such combining algorithm might not give the optimal sample complexity. The adaptiveness of such combining procedure splits the rounds into rounds used by the top-$k$ algorithm and the rounds used by the sorting algorithm. As we will see later, this is not the optimal way to solve sorted top-$k$ in bounded number of rounds.

In this paper, we show optimal sample complexity bounds (up to poly logarithmic factors) of sorted top-$k$ in $r$ rounds for any constant $r$, as shown in Table \ref{tab:stopk_noiseless}. Our bounds are tight up to constant factors when $r= 1$ or $2$.  

\begin{table}[H]
\label{tab:stopk_noiseless}
\begin{center}
  \begin{tabular}{| c  | c | c| }
    \hline
     Number of Rounds  & Upper Bound & Lower Bound \\ \hline
    1-round & $O(n^2)$ & $\Omega(n^2)$ \\ \hline
    2-round & $O(n \sqrt{k}+ n^{4/3} )$ & $\Omega( n \sqrt{k} +n^{4/3} )$ \\ \hline
    ($r\geq 3$)-round &$O((n^{2/r} k^{(r-1)/r}+n )polylog(n))$ & $\Omega(n^{2/r} k^{(r-1)/r}+n)$ \\ \hline
  \end{tabular}
\end{center}
\caption{Sorted top-$k$ with noiseless comparisons}
\end{table}

We further extend our results to sorted top-$k$ in the noisy case where each comparison is correct with probability $2/3$. This is a very simple and basic noise model. As shown in Table \ref{tab:stopk_noisy}, we get tight bounds (up to a constant factor) when $r=1$ or 2. Compared with the sample complexity in the noiseless case, the sample complexity in the noisy case just has an extra $\log(k)$ factor when $r=1$ or 2. 

\begin{table}[H]
\label{tab:stopk_noisy}
\begin{center}
  \begin{tabular}{| c  | c | c| }
    \hline
     Number of Rounds  & Upper Bound & Lower Bound \\ \hline
    1-round & $O(n^2 \log(k))$ & $\Omega(n^2\log(k))$ \\ \hline
    2-round &$O((n \sqrt{k}+ n^{4/3})\log(k))$ & $\Omega((n \sqrt{k}+ n^{4/3})\log(k))$ \\ \hline
    ($r\geq 3$)-round &$O((n^{2/r} k^{(r-1)/r}+n)polylog(n))$ &$\Omega(n^{2/r} k^{(r-1)/r}+n)$ \\ \hline
  \end{tabular}
\end{center}
\caption{Sorted top-$k$ with noisy comparisons}
\end{table}

Our techniques also give new results for top-$k$ and sorting in the noisy case. Our sorted top-$k$ algorithms are based on our top-$k$ algorithms. For top-$k$ in the noisy case, we show the tight sample complexity bounds are $\Theta(n^2)$ for $r=1$ and $\Theta(n^{4/3})$ for $r=2$. 
On the other hand, sorting is a sub-case of sorted top-$k$ by picking $k=n$.  Our sorted top-$k$ results imply tight bounds for sorting in the noisy case: $\Theta(n^2 \log(n))$ for $r=1$ and $\Theta(n^{3/2} \log(n))$ for $r=2$.

\subsection{Related Work}
Sorted top-$k$ is first discussed in \cite{Chambers71} and is referred as ``partial sorting''.

The parallel comparison model is introduced by \cite{Valiant75}. Two problems which are related to sorted top-$k$ have been widely studied in this model: sorting \cite{HaggkvistH81, AjtaiKS83,BollobasT83,Kruskal83,Leighton84,BollobasH85,Alon85,AlonAV86,AzarV87,Pippenger87,AlonA88b,AlonA88a,Akl90} and selection \cite{Valiant75,Reischuk81,AjtaiKSS86,Pippenger87,AlonA88b,AlonA88a,AzarP90,BollobasB90}.
Selection (i.e. finding the item of rank exactly $k$) has been shown to be very similar to top-$k$ in \cite{BMW16}.

The noisy comparison model was introduced by \cite{FeigeRPU94}. Recently, there are several work studying top-$k$ with noisy comparisons in bounded number of rounds \cite{BMW16,AgarwalAAK17,CohenMM18}. \cite{BMW16} shows that the sample complexity of top-$k$ in the noisy case is $\tilde{\Theta}(n)$ when $r=3$ and $\Theta(n\log(n))$ when $r\geq 4$ and $\max(k,n-k) = \Omega(n)$. \cite{CohenMM18} gives tight sample complexity bound for general $k$ and $r>4$. The sample complexity of top-$k$ in the noisy case for $r=1$ and 2 is not addressed in previous work. 

Rank aggregation with noisy comparisons is also widely studied without the round constraint and with various noise models: \cite{Mathieu07,AilonCN08,BravermanM08,BravermanM09,Ailon11,Jamieson11,LuB11,MakarychevM13,WauthierJJ13,RajkumarA14,ChenS15,ShahW15,MohajerS16,NegahbanOS17,CGMS17,ShahBGW17,SuhTZ17,CLM18}.

\section{Model and Preliminaries}

We consider the sorted top-$k$ problem together with two related problems: sorting and top-$k$. In these problems, there is a set of $n$ items $N$ with an underlying order and the goals are different:
\begin{itemize}
\item \textbf{Sorted top-$k$:} output the sorted list of $k$ items with highest ranks.
\item \textbf{Sorting:} output the ranks of all items.
\item \textbf{Top-$k$:} output the set of $k$ items with highest ranks.
\end{itemize}

Algorithms are allowed to make pairwise comparisons. And we want the algorithm to minimize the sample complexity, i.e. the total number of comparisons. We have two cases, with respect to comparisons: the noiseless case and the noisy case. In the noiseless case, comparisons results are always consistent with the underlying order. In the noisy case, each pairwise comparison is correct (consistent with the underlying order) with some constant probability $>1/2$ independently. Without loss of generality, we assume each comparison is correct with probability $2/3$ independently. 

We consider algorithms with bounded number of rounds. In each round, an algorithm needs to perform all comparisons simultaneously. We use $r$ to denote the number of rounds and we only consider cases when $r$ is a fixed constant. 

We allow algorithms to use randomness. In the noiseless case, the algorithm needs to be always correct and the sample complexity is the expected total number of comparisons. In the noisy case, because of the noise, no algorithms can always be correct. The algorithm needs to be correct with probability $\geq 2/3$ and the sample complexity is the worst-case total number of comparisons. Notice that the requirement in the noiseless case is stronger as the always correct algorithm with expected number of comparisons $s$ can be easily made into an algorithm which succeeds with probability $2/3$ and worst-case number of comparisons $O(s)$ by halting the algorithm when making too many comparisons.

\section{Main Results and Proof Overviews}
In this section, we show our main results and give overviews of our proof techniques.

\subsection{Sorted Top-$k$ in the Noiseless Case}
In this sub-section, we show our results for sorted top-$k$ in the noiseless case. All the detailed discussions and proofs can be found in Section \ref{sec:stopk}.

When we only have 1 round (i.e. $r=1$), it is not hard to show that comparing all pairs of items using $\Theta(n^2)$ comparisons gives optimal sample complexity (up to a constant factor). We formally discuss this in the beginning of Section \ref{sec:stopk}. For $r\geq 2$, we have the following two main theorems for upper and lower bounds.

\begin{theorem}
\label{thm:stopk_ub}
For $r \geq 3$, there exists an $r$-round algorithm that solves sorted top-$k$ with $\tilde{O}(n^{2/r}k^{(r-1)/r}+ n))$ comparisons in expectation. There exists a $2$-round algorithm that solves sorted top-$k$ with $O(n\sqrt{k}+ n^{4/3})$ comparisons in expectation.
\end{theorem}

The main idea of the algorithm in Theorem \ref{thm:stopk_ub} is to use ``pivot items''. 
These pivot items are compared to all items. From these comparisons, we learn not only their ranks but also which items rank between two pivot items. After that, items are partitioned into chunks and we just need to solve sub problems inside chunks. See Figure \ref{fig:pivot} for a graphical view of pivot items.
\begin{figure}[H]
\centering
\begin{tikzpicture} [scale = 0.9]
\def\x{15}
\draw [line width = 1mm] (0,0) -- (\x,0);
\foreach \y in {0,1,2,3,4,5,6,7,8,9,10,11,12,13,14,15}
{
\filldraw [fill=red!20!white] (\y,0) circle (0.3cm);
} 

\foreach \y in {3,6,10,14}
{
\filldraw [fill=red!100!white] (\y,0) circle (0.4cm);
}


\def\e{0.1}
\node (a) at (8,-1) {chunk};
\draw[|-,line width = 0.5mm] (6+\e,-1) -- (a);
\draw[-|,line width =0.5mm] (a)  -- (10-\e,-1);


\filldraw [fill=red!100!white] (1,1.5) circle (0.4cm);
\node at (2.6,1.5) {pivot items};
\filldraw [fill=red!20!white] (5,1.5) circle (0.3cm);
\node at (6.6,1.5) {other items};

\node at (14.5,1) {rank order};
\draw[->] (13.5,0.7)--(15.5,0.7);

\end{tikzpicture}
\caption{Pivot items.}
\label{fig:pivot}
\end{figure}
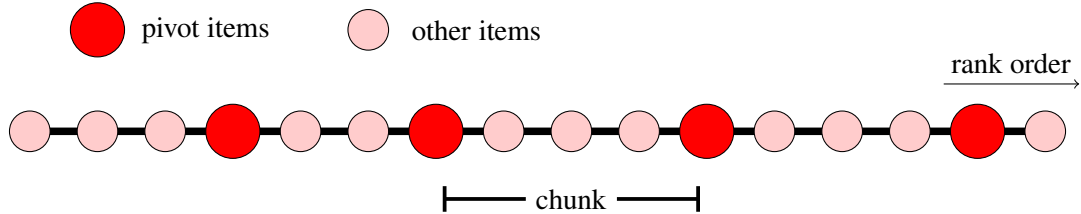

Suppose we plan to use $\Theta(\alpha n)$ comparisons. The most naive way of using pivot items is to pick $\alpha$ pivot items at random in the first round and compare them to all items in the same round. After this round, we will be left with chunks of items partitioned by pivot items. Since now we know the ranks of pivot items, we know which chunks have top-$k$ items and we only need to care about these chunks. We use the remaining $r-1$ round to run the $(r-1)$-round sorting algorithm of \cite{AlonAV86} in each such chunk in parallel. This approach with proper setting of $\alpha$ matches the optimal sample complexity bound (up to a constant factor) if $r=2$ or $k$ is larger than the expected chunk size (i.e. $k = \Omega(n/\alpha)$). It is formally described in Algorithm \ref{alg:rsorted1} in Section \ref{sec:stopk}. See also Figure \ref{fig:pivot1} for a graphical view of the algorithm.

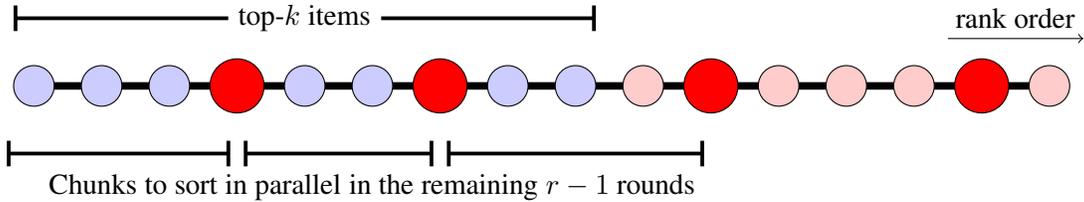
\begin{figure}[H]
\centering
\begin{tikzpicture} [scale = 0.9]
\def\x{15}
\draw [line width = 1mm] (0,0) -- (\x,0);
\foreach \y in {9,10,11,12,13,14,15}
{
\filldraw [fill=red!20!white] (\y,0) circle (0.3cm);
} \foreach \y in {0,1,2,3,4,5,6,7,8}
{
\filldraw [fill=blue!20!white] (\y,0) circle (0.3cm);
} 

\foreach \y in {3,6,10,14}
{
\filldraw [fill=red!100!white] (\y,0) circle (0.4cm);
}

\def\e{0.1}
\node (a) at (4,1) {top-$k$ items};
\draw[|-,line width = 0.5mm] (-0.3,1) -- (a);
\draw[-|,line width =0.5mm] (a)  -- (8.3,1);

\node at (14.5,1) {rank order};
\draw[->] (13.5,0.7)--(15.5,0.7);

\draw[|-|,line width = 0.5mm] (-0.5+\e,-1)--(3-\e,-1);
\draw[|-|,line width = 0.5mm] (3+\e,-1)--(6-\e,-1);
\draw[|-|,line width = 0.5mm] (6+\e,-1)--(10-\e,-1);
\node at (5,-1.5) {Chunks to sort in parallel in the remaining $r-1$ rounds};
\end{tikzpicture}
\caption{Noiseless sorted top-$k$ algorithm when $k$ is large.}
\label{fig:pivot1}
\end{figure}

However, when $r>3$ and $k$ is small enough so that the first chunk is likely to have size much larger than $k$ (see Figure \ref{fig:pivot2}), the above approach becomes sub-optimal. At a high level,  the reason is that the random pivot items chosen in the first round are not good enough to partition the top-$k$ items into small chunks. Therefore, instead of running the sorting algorithm on the first chunk in the remaining $r-1$ rounds, we spend one more round (round 2) to compare items to more ``accurate'' pivot items, partition them into smaller chunks and sort each chunk in the remaining $r-2$ rounds. These new pivot items are better than the random pivot items for two reasons: (i) We can spend some comparisons in the first round to choose these pivot items. So they have better structural guarantees than the random pivot items. In particular, we extend the algorithmic technique of \cite{BMW16} to pick pivot items which are roughly $\Theta\left(\sqrt{\frac{n}{\alpha}}\right)$ apart. (ii) Since these pivot items are compared to other items in the second round, we can use comparison results of the random pivot items. Knowing the fact that all top-$k$ items are in the first chunk partitioned by the random pivot items, we just need to compare new pivot items to items in that chunk. This is important for getting good sample complexity. 
The whole process of this paragraph is formally described in Algorithm \ref{alg:rsorted2} in Section \ref{sec:stopk}. See also Figure \ref{fig:pivot2} for a graphical view of the algorithm.
\begin{figure}[H]
\centering
\begin{tikzpicture}  [scale = 0.9]
\def\x{15}
\draw [line width = 1mm] (0,0) -- (\x,0);
\foreach \y in {2,3,4,5,6,7,8,9,10,11,12,13,14,15}
{
\filldraw [fill=red!20!white] (\y,0) circle (0.3cm);
} \foreach \y in {0,1}
{
\filldraw [fill=blue!20!white] (\y,0) circle (0.3cm);
} 

\foreach \y in {3,6,10,14}
{
\filldraw [fill=red!100!white] (\y,0) circle (0.4cm);
}

\def\e{0.1}
\node (a) at (0.5,1) {top-$k$};
\draw[|-,line width = 0.5mm] (-0.3,1) -- (a);
\draw[-|,line width =0.5mm] (a)  -- (1.3,1);

\node at (14.5,1) {rank order};
\draw[->] (13.5,0.7)--(15.5,0.7);

\draw[|-|,line width = 0.5mm] (-0.5+\e,-1)--(3-\e,-1);
\node at (1.25,-1.5) {First chunk};

\draw [line width = 1mm] (0,-3) -- (\x+0.7,-3);
\foreach \y in {9,11,13,14,15}
{
\filldraw [fill=red!20!white] (\y,-3) circle (0.3cm);
} \foreach \y in {0,1,3,4,6,8}
{
\filldraw [fill=blue!20!white] (\y,-3) circle (0.3cm);
} 
\filldraw [fill=red!100!white] (15,-3) circle (0.4cm);
\foreach \y in {2,5,7,10,12}
{
\filldraw [fill= white] (\y,-3) circle (0.4cm);
\filldraw [pattern = crosshatch, pattern color=red!50!black] (\y,-3) circle (0.4cm);
}

\draw  [dashed, line width = 0.5mm] (-0.5,-0.5)--(3.5,-0.5)--(3.5,0.5)--(-0.5,0.5)--cycle;

\draw  [dashed, line width = 0.5mm] (-0.5,-3.5)--(15.5,-3.5)--(15.5,-2.5)--(-0.5,-2.5)--cycle;

\draw [->, line width = 0.5mm] (2.5,-0.5)--(7.5,-2.5);
\node at (6.5,-1.5) {zoom in};

\filldraw [pattern = crosshatch, pattern color=red!50!black] (9.7,-1.5) circle (0.4cm);
\node at (12.8,-1.5) {: more ``accurate'' pivot items};
\draw[|-|,line width = 0.5mm] (-0.5+\e,-4)--(2-\e,-4);
\draw[|-|,line width = 0.5mm] (2+\e,-4)--(5-\e,-4);
\draw[|-|,line width = 0.5mm] (5+\e,-4)--(7-\e,-4);
\draw[|-|,line width = 0.5mm] (7+\e,-4)--(10-\e,-4);
\node at (5,-4.5) {Chunks to sort in parallel in the remaining $r-2$ rounds};
\end{tikzpicture}

\caption{Noiseless sorted top-$k$ when $k$ is small.}
\label{fig:pivot2}
\end{figure}
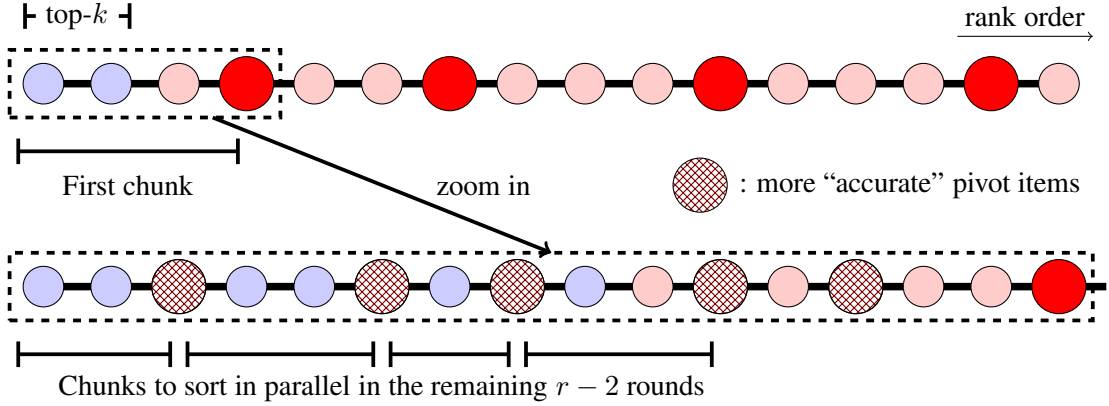

In Section \ref{sec:stopk}, we combine the above two approaches to prove Theorem \ref{thm:stopk_ub}. Both approaches use pivot items and then bounded-round sorting. Although the sample complexity keeps decreasing when we increase the number of rounds, we have at most 2 rounds that are different from sorting no matter how large the total number of rounds is. One may wonder why we don't use more rounds before we apply bounded-round sorting. At a high level, the reason is that what we do before sorting is more similar to a top-$k$ algorithm and more than 3 rounds of interaction do not help much with the sample complexity for top-$k$, e.g. \cite{BMW16} shows a 3-round noiseless top-$k$ algorithm with nearly optimal sample complexity $O(n \cdot polylog(n))$.

\begin{theorem}
\label{thm:stopk_lb}
  For $r \geq 3$, $r$-round algorithm needs $\Omega(n^{2/r} k^{(r-1)/r} +n)$ comparisons in expectation to solve sorted top-$k$. Any $2$-round algorithm needs $\Omega(n \sqrt{k}+ n^{4/3})$ comparisons in expectation to solve sorted top-$k$.
\end{theorem}

Theorem \ref{thm:stopk_lb} gives matching (up to constant or polylog factors) lower bounds compared to upper bounds in Theorem \ref{thm:stopk_ub}. The start point of the proof is to reduce from top-$k$ or sorting to sorted top-$k$. Indeed, sorted top-$k$ is no easier than sorting $k$ items or finding top-$k$ items over $n$ items. However, this reduction is not good enough to give us tight lower bounds. 

Let us we go back to our sorted top-$k$ algorithm in Theorem \ref{thm:stopk_ub} and compare how it is different from an algorithm which is given the set of top-$k$ items and just sorts these $k$ items. The main difference is that our sorted top-$k$ algorithm spends a big fraction of comparisons in the first one or two rounds on items which are not top-$k$ items. These comparisons are not useful for sorting the top-$k$ items. Moreover, we can show that, without knowing the set of top-$k$ items, not only our algorithm but also any other algorithms will make a good amount of comparisons outside top-$k$ in the first one or two rounds. For example, it is not hard show that in expectation at most $O(k^2/n^2)$ fraction of first-round comparisons are between two items in top-$k$. The argument for the second round is more complicated. 

This is the critical point of our proof. Now we know that sorted top-$k$ is no easier than sorting $k$ items with unbalanced number of comparisons in rounds (fewer comparisons in the first one or two rounds). In the rest of proof, we adapt the lower bound of bounded-round sorting (Theorem 2.1 of \cite{AlonA88b}) to our unbalanced setting. For details, see Section \ref{sec:stopk_lb}.

\subsection{Warm-up: Top-1 in the Noisy Case}
Now we proceed to the noisy case. First of all, one could easily adapt a noiseless algorithm into the noisy case by repeating each comparison $\Theta(\log(n))$ times and use union bound in the analysis. So the interesting question here is whether the sample complexity gets an extra $\Theta(\log(n))$ factor or not or something in-between, when we transition from the noiseless case to the noisy case.


In this sub-section, we show an 1-round algorithm for finding top-1 in the noisy case with $O(n^2)$ comparisons. The sample complexity does not get a $\Theta(\log(n))$ blow-up in the noisy case. This algorithm is simpler than and different from our 1-round algorithms for top-$k$ and sorted top-$k$ in the noisy case. We offer it here as a warm-up for the noisy case. 

Without loss of generality, we assume $n$ is a power of 2. If $n$ is not a power of 2, we could add fewer than $n$ dummy items to make the total number of items a power of 2. This only increase the number of comparisons by a constant factor.

In Algorithm \ref{alg:1top1}, we show our main recursive procedure of finding the top-1 item in some set $S$ of size $s$. We will show by induction in Lemma \ref{lem:top1} that it uses $O(s^2\log(1/\delta))$ comparisons and succeeds with probability at least $1-\delta$. If we run $FindMax(N, n, 1/9)$, we will get an algorithm for finding top-1 within $n$ items with probability at least $8/9$ in the noisy case using $O(n^2)$ comparisons. Notice that although the procedure is defined recursively, no pair of items in an comparison depends on other comparisons' results. So all the comparisons can be done in 1-round.

This recursive procedure basically partitions set $S$ into two set $S_1$ and $S_2$ of equal sizes and then find the max in each set recursively: item $i^*$ and item $j^*$. After that it compares $i^*$ and $j^*$ some times to find the max of this two. In order to make all comparisons in 1 round, we actually compare all pairs of items $(i,j)$ for $i\in S_1$ and $j\in S_2$.

In order to make this recursive procedure to succeed with probability $1-\delta$, we want that the following three steps all succeed with probability $1-\delta/3$ and we take a union bound: (1) finding the max of $S_1$: item $i^*$ (2) finding the max of $S_2$: item $j^*$ (3) finding the max of $i^*$ and $j^*$. As you will see in the proof, the critical point of the argument is to show that the growth in the success probability (from $1-\delta$ to $1-\delta/3$) has much less effect on the sample complexity compared with the decrease of the set size (from $|S|$ to $|S_1|=|S_2|=|S|/2$).

\begin{algorithm}[t]
    \caption{FindMax($S$, $s$, $\delta$)}
    	\label{alg:1top1}
    \begin{algorithmic}[1]
	\IF {$s = 1$}
		\STATE Return the single item in $S$.
	\ELSE
    	\STATE Partition $S$ arbitrarily into set $S_1$ and $S_2$ of equal sizes, i.e. $|S_1|=|S_2| = s/2$, $S_1 \cap S_2 =\emptyset$ and $S_1 \cup S_2 = S$.  
		\STATE For each item $i \in S_1$ and $j \in S_2$, compare them $100 \log(1/\delta)$ times.
		\STATE $i^* \leftarrow \text{FindMax}(S_1,s/2,\delta/3)$.
		\STATE $j^* \leftarrow \text{FindMax}(S_2, s/2, \delta/3)$.
		\STATE Return $i^*$ if item $i^*$ wins the majority of comparisons between item $i^*$ and item $j^*$. Return $j^*$ otherwise.
	\ENDIF
    \end{algorithmic}
\end{algorithm}

\begin{lemma}
\label{lem:top1}
Let $\delta \leq 1/9$. FindMax($S$, $s$, $\delta$) and its descendants use at most $100n^2\log(1/\delta))$ comparisons. FindMax($S$, $s$, $\delta$) succeeds to output the top-1 in $S$ with probability at least $1-\delta$.
\end{lemma}

\begin{proof}
We know that $|S| = s$ is always a power of 2, i.e. $s= 2^t$. We prove the lemma by induction on $t$. The base case $t = 0$ is trivial. 

Let's assume the lemma is true for $t-1$, let's consider the case for $t$. By induction hypothesis, we know the number of comparisons in $\text{FindMax}(S_1,s/2,\delta/3)$ and its descendants is at most $100 (s/2)^2 \log(3/\delta)$. Same for  $\text{FindMax}(S_2,s/2,\delta/3)$. Therefore, the total number of comparisons used by FindMax($S$, $s$, $\delta$) and its descendants is 
\begin{align*}
&2 \cdot 100 (s/2)^2 \log(3/\delta) + 100 \log(1/\delta) \cdot (s/2)^2 \\
= &100 \cdot \frac{s^2}{4} \cdot \log(1/\delta) \left(2 +\frac{ 2\log(3)}{\log(1/\delta)} +1\right)\\
 \leq& 100 s^2 \log(1/\delta). 
\end{align*}

In the case that $i^*$ is the top-1 of $S_1$, $j^*$ is the top-1 of $S_2$ and the majority of comparisons between $i^*$ and $j^*$ is consistent with their true ordering, FindMax($S$, $s$, $\delta$) succeeds to output the top-1 in $S$. By induction hypothesis, each of the first two events happens with probability at least $1-\delta/3$. By Chernoff bound, the third event happens with probability at least $1- \exp\left((1/4)^2 \cdot (1/2) \cdot (2/3) \cdot 100 \log(1/\delta) \right) \geq 1-\delta /3$. Therefore, by union bound,  FindMax($S$, $s$, $\delta$) succeeds to output the top-1 in $S$ with probability at least $1-\delta$.
\end{proof}

\subsection{Top-$k$ in the Noisy Case}
In this sub-section, we discuss top-$k$ in the noisy case. All the detailed discussions and proofs about this sub-section can be found in Section \ref{sec:topk_n}.

As discussed in the related work, top-$k$ in the noisy case has been studied in prior work when $r \geq 3$. Nothing is known when $r =1$ or $2$. On the other hand, if we go back to the noiseless case, it has been shown in prior work that the sample complexity of top-$k$ is $\Theta(n^2)$ for $r=1$ and $\Theta(n^{4/3})$ for $r=2$.

We show top-$k$ algorithms in the noisy case in Theorem \ref{thm:topk_n} for $r=1$ or $2$. These upper bounds are tight up to a constant factor as they even match the lower bounds in the noiseless case. In other words, for top-$k$ in 1 round or 2 rounds, the sample complexity does not get an extra $\Theta(\log(n))$ factor when we go from the noiseless case to the noisy case. 

\begin{theorem}
\label{thm:topk_n}
For top-$k$ in the noisy case, there is an 1-round algorithm with sample complexity $O(n^2)$ and a 2-round algorithm with sample complexity $O(n^{4/3})$. 
\end{theorem}

Our 1-round algorithm starts by the simple idea of comparing two items $\Theta(\log(n))$ times. The majority of these comparison is consistent with the true ordering with probability $1 - 1/poly(n)$. By taking a union bound later in the analysis, $\Theta(\log(n))$ noisy comparisons between the same pair of two items can be considered as one noiseless comparison between them.

Since we plan to repeat each comparison $\Theta(\log(n))$ times and we have only $O(n^2)$ comparisons, we cannot compare all pairs of items. So we use pivot items again. We pick $\Theta(n/\log(n))$ pivot items at random and compare them to all items 
$\Theta(\log(n))$ times. We can partition items into chunks. For items rank before or after the chunk which contains the $k$-th item, we can easily classify them as in the top-$k$ or outside top-$k$. For items inside this chunk, we don't know which ones are in top-$k$. Since each chunk has $\Theta(\log(n))$ items in expectation, the number of such items is $\Theta(\log(n))$ in expectation.

How do we deal with these $\Theta(\log(n))$ items? We use more random pivot items. We pick $\Theta(n/\log\log(n))$ random pivot items and further partition items into chunks of size $\Theta(\log\log(n))$ in expectation. We call these new pivot items as second-level pivot items and previous pivot items as first-level pivot items (see Figure \ref{fig:topk1}). Here comes to the critical point of the argument: since second-level pivot items are only used to partition $\Theta(\log(n))$ items and in the analysis we are taking union bound over $polylog(n)$ events, we don't need to repeat the comparison between each pair $\Theta(\log(n))$ times. Instead, we just need to repeat each comparison $\Theta(\log\log(n))$ times. And our total number comparisons will still be $O(n^2)$. 

Finally we generalize this idea to have $\log^*(n)$ levels of pivot items and we can classify all items into top-$k$ or bottom-$(n-k)$. Moreover, although these pivot items are divided into different levels, they are all chosen at random and compared to all items. So all the comparisons can be placed in a single round. The whole algorithm is formally described in Algorithm \ref{alg:1topk} in Section \ref{sec:topk_n}. See also Figure \ref{fig:topk1} for a graphical view of the algorithm.


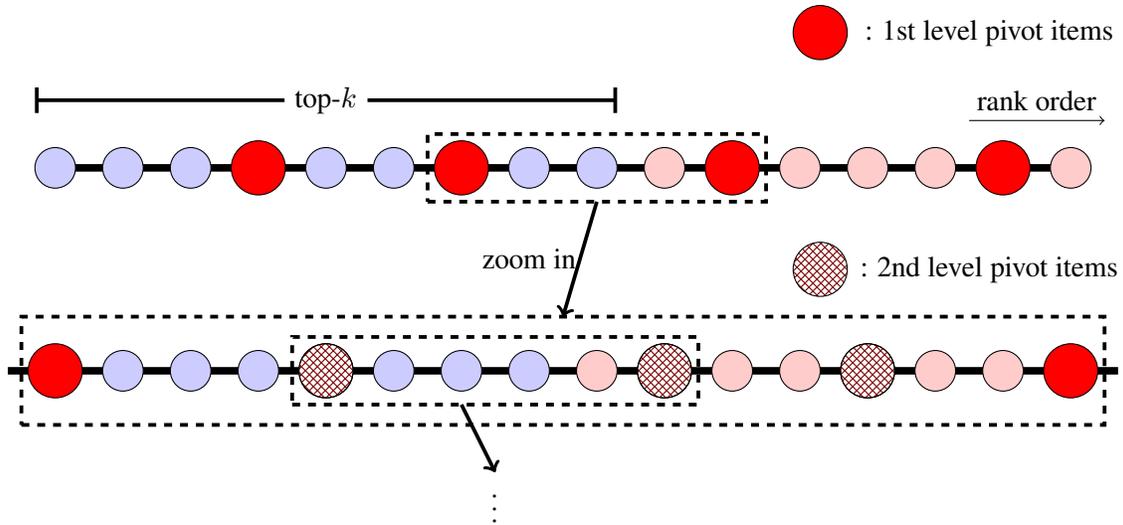
\begin{figure}[H]
\centering
\begin{tikzpicture}  [scale = 0.9]
\def\x{15}
\draw [line width = 1mm] (0,0) -- (\x,0);
\foreach \y in {9,10,11,12,13,14,15}
{
\filldraw [fill=red!20!white] (\y,0) circle (0.3cm);
} \foreach \y in {0,1,2,3,4,5,6,7,8}
{
\filldraw [fill=blue!20!white] (\y,0) circle (0.3cm);
} 

\foreach \y in {3,6,10,14}
{
\filldraw [fill=red!100!white] (\y,0) circle (0.4cm);
}

\def\e{0.1}
\node (a) at (4,1) {top-$k$};
\draw[|-,line width = 0.5mm] (-0.3,1) -- (a);
\draw[-|,line width =0.5mm] (a)  -- (8.3,1);

\node at (14.5,1) {rank order};
\draw[->] (13.5,0.7)--(15.5,0.7);


\draw [line width = 1mm] (-0.7,-3) -- (\x+0.7,-3);
\foreach \y in {8,9,10,11,12,13,14,15}
{
\filldraw [fill=red!20!white] (\y,-3) circle (0.3cm);
} \foreach \y in {0,1,2,3,4,5,6,7}
{
\filldraw [fill=blue!20!white] (\y,-3) circle (0.3cm);
} 
\filldraw [fill=red!100!white] (15,-3) circle (0.4cm);
\filldraw [fill=red!100!white] (0,-3) circle (0.4cm);
\foreach \y in {4,9,12}
{
\filldraw [fill= white] (\y,-3) circle (0.4cm);
\filldraw [pattern = crosshatch, pattern color=red!50!black] (\y,-3) circle (0.4cm);
}

\draw  [dashed, line width = 0.5mm] (5.5,-0.5)--(10.5,-0.5)--(10.5,0.5)--(5.5,0.5)--cycle;

\draw  [dashed, line width = 0.5mm] (-0.5,-3.8)--(15.5,-3.8)--(15.5,-2.2)--(-0.5,-2.2)--cycle;
\draw  [dashed, line width = 0.5mm] (3.5,-3.5)--(9.5,-3.5)--(9.5,-2.5)--(3.5,-2.5)--cycle;
\draw [->, line width = 0.5mm] (8,-0.5)--(7.5,-2.2);
\draw [->, line width = 0.5mm] (6,-3.5)--(6.5,-4.5);
\node [rotate = 90] at (6.5,-5) {$\cdots$};
\node at (7,-1.35) {zoom in};

\filldraw [pattern = crosshatch, pattern color=red!50!black] (11.3,-1.5) circle (0.4cm);
\node at (13.8,-1.5) {: 2nd level pivot items};

\filldraw [fill=red!100!white] (11.3,2) circle (0.4cm);
\node at (13.8,2) {: 1st level pivot items};
\end{tikzpicture}

\caption{The noisy 1-round top-$k$ algorithm.}
\label{fig:topk1}
\end{figure}

Now we proceed to describing our 2-round top-$k$ algorithm in the noisy case. It is the most sophisticated algorithm in this paper. We are going use $O(n^{4/3})$ comparisons. 

It would be good to first understand the 2-round top-$k$ algorithm in the noiseless case with $O(n^{4/3})$ comparisons. The idea is quite simple: we pick $n^{1/3}$ random pivot items in the first round and partition items into chunks of size $\Theta(n^{2/3})$ in expectation. And in the second round, we just need to focus on the chunk containing the $k$-th item. It has size $\Theta(n^{2/3})$ in expectation and we can just compare all pairs of items in this chunk.

Now in the noisy case, how do we still use only $O(n^{4/3})$ comparisons to find top-$k$ in 2 rounds? Repeating each comparison $\Theta(\log(n))$ times does not seem to work since it reduces the number of random pivot items to $\Theta(n^{1/3} / \log(n))$ and we will leave a chunk of too many items (i.e. $\Theta(n^{2/3} \log(n))$ items) to the second round. 

In our 2-round algorithm, we still use $\Theta(n^{1/3})$ random pivot items in the first round. We can only compare them to all items constant times as we only have $O(n^{4/3})$ comparisons in total. We partition items into chunks as following (see also Figure \ref{fig:topk2} for a graphical view). We first put pivot items in the order of their ranks. We get this order correctly with probability $1-1/poly(n)$ after the first round by having $\Theta(\log(n))$ comparison between each pair of pivot items in parallel with other comparisons. For each item $i$, we keep a counter and compare it to pivot items one by one. The counter is increased by one if the pivot item wins the majority of comparisons with item $i$ and decreased by one otherwise. In the end, we put the item into the chunk next to the pivot item where its counter reaches its maximum. The analysis of this process is similar to a biased random walk. Notice that item $i$'s counter has higher chance of increasing before it reaches its actual chunk and it has higher chance of decreasing after it reaches its actual chunk. We can show that although we may fail to put item $i$ into its actual chunk, the probability it is placed $l$ chunks away from its actual chunk can be bounded by $\exp(-\Omega(l))$.


\begin{figure}[H]
\centering
\begin{tikzpicture}  [scale = 0.9]
\def\x{15}
\draw [line width = 1mm] (0,0) -- (\x,0);

\filldraw [fill=red!20!white] (0,1) circle (0.3cm);

\filldraw [fill=red!20!white] (8,0) circle (0.3cm);
\foreach \y in {3,6,10,13}
{
\filldraw [fill=red!100!white] (\y,0) circle (0.4cm);
}

\def\e{0.1}

\node at (14.5,1) {rank order};
\draw[->] (13.5,0.7)--(15.5,0.7);


\draw [->, line width = 0.5mm] (0.3,1) to [out=20,in=120] (2.9,0.4);

\draw [->, line width = 0.5mm] (3.1,0.4) to [out=60,in=120] (5.9,0.4);

\draw [->, line width = 0.5mm] (6.1,0.4) to [out=60,in=120] (9.9,0.4);

\draw [->, line width = 0.5mm] (10.1,0.4) to [out=60,in=120] (12.9,0.4);
\node at (3,1) {\huge $\color{green}\checkmark$};
\node at (6,1) {\huge $\color{green}\checkmark$};
\node at (10,1) {\huge $\mathbin{\tikz [x=1.4ex,y=1.4ex,line width=.2ex, red] \draw (0,0) -- (1,1) (0,1) -- (1,0);}$};
\node at (13,1) {\huge $\mathbin{\tikz [x=1.4ex,y=1.4ex,line width=.2ex, red] \draw (0,0) -- (1,1) (0,1) -- (1,0);}$};

\node at (3,1.7) {1};
\node at (6,1.7) {2};
\node at (10,1.7) {1};
\node at (13,1.7) {0};

\node at (0.5,1.7) {counter:};
\end{tikzpicture}

\caption{The first round of the noisy 2-round top-$k$ algorithm.}
\label{fig:topk2}
\end{figure}
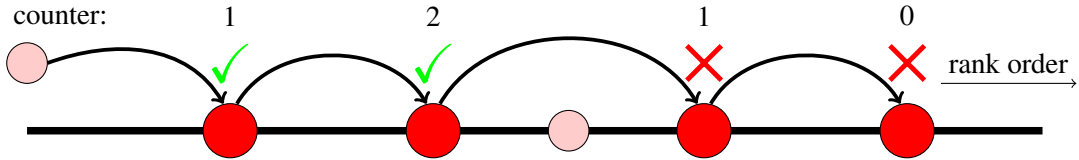

After the first round, we partition items into chunks of size $\Theta(n^{2/3})$ in expectation. As discussed above, this partition is not perfectly correct but items won't be placed too far away from their actual chunks. If the partition is perfectly correct, then we can directly use our previous 1-round top-$k$ algorithm described above as a blackbox to deal with the chunk containing the $k$-th item. But since the partition is not perfectly correct, we have to modify the 1-round algorithm to use in the second round of our 2-round algorithm. The whole algorithm is formally described in Algorithm \ref{alg:2topk} in Section \ref{sec:topk_n}.

\subsection{Sorted Top-$k$ in the Noisy Case}
In this sub-section, we show our results for sorted top-$k$ in the noisy case. All the detailed discussions and proofs can be found in Section \ref{sec:stopk_n}.  In the noisy case for round number $r \geq 3$, as described in the previous sub-section, we can adapt our noiseless algorithm into a noisy algorithm with sample complexity $O((n^{2/r} k^{(r-1)/r}+n)polylog(n))$. For $r= 1,2$, we show tight (up to a constant factor) sample complexity bounds in Theorem \ref{thm:stopk_n}.

\begin{theorem}
\label{thm:stopk_n}
The sample complexity of sorted top-$k$ in the noisy case is $\Theta(n^2 \log(k))$ for $r=1$ and $\Theta((n\sqrt{k} + n^{4/3}) \log(k))$ for $r=2$. 
\end{theorem}

Both of our sorted top-$k$ algorithms in Theorem \ref{thm:stopk_n} are based on our top-$k$ algorithms. Our 1-round sorted top-$k$ algorithm is relatively simple given our 1-round top-$k$ algorithm. We just compare all pairs $\Theta(\log(k))$ times and also runs the 1-round algorithm for top-$k$ of Theorem \ref{thm:topk_n} in the same round. After we make the comparisons, we learn the set of top-$k$ items. The majority of $\Theta(\log(k))$ comparisons between each pair is consistent to the actual rankings with probability $1-1/poly(k)$. Since we only focus on $k$ items, we can take the union bound to show our algorithm is correct with probability at least $2/3$. For details, see Algorithm \ref{alg:1stopk_noisy}.


Interestingly, although sorted top-$k$ is no easier than top-$k$, getting tight bounds of sorted top-$k$ could be easier. Our 2-round sorted top-$k$ algorithm is much simpler than our 2-round top-$k$ algorithm, and it only depends on our 1-round top-$k$ algorithm. When $k$ is not tiny ($\geq n^{1/10}$), since $\log(k) = \Theta(\log(n))$, we just use the sorted top-$k$ algorithm in the noiseless case (Theorem \ref{thm:stopk_ub}) and repeat each comparison $\Theta(\log(n))$ times. When $k$ is tiny  ($< n^{1/10}$), we partition all items into random groups; find the top-1 of each group in the first round (using the 1-round top-$k$ algorithm of Theorem \ref{thm:topk_n}) and then find the sorted top-$k$ of all these top-1's in the second round. For details, see Algorithm \ref{alg:2stopk_noisy}.


Now we start to describe how we prove the matching lower bounds. We start with the 1-round lower bound. The main idea of the lower bound is to show that if an algorithm does not make enough comparisons in one round, there must exist $\Omega(k)$ pairs of items who have consecutive ranks and are in top-$k$, such that they are compared fewer than $\log(k)/2$ times. 
For any one such pair of items, if we just swap their ranks, the order of items in top-$k$ changes and we can show that the chance of seeing the same comparison result would at most decrease by a factor of $2^{\log(k)/2}$. As long as the number of such pairs is much larger than this factor, we can show that an 1-round algorithm with not enough comparisons outputs incorrectly with large probability. For details, see Lemma \ref{lem:1stopk_noisy_lb}.

For the 2-round lower bound, we still want to show that if an algorithm does not make enough comparisons in two rounds, there must exist enough pairs of items who have consecutive ranks and are in top-$k$, such that they are compared fewer than $c\cdot \log(k)$ times for some small constant $c$. The proof is more complicated as the a 2-round algorithms have adaptiveness, i.e. which items are compared in the second round depend on the comparison results of the first round. 
The main idea of the proof is to show that bounded amount of first round comparisons won't be helpful to figure out which items are consecutively ranked. 
We explain proof steps in the case when $k > n^{2/3}$ here. We divide top-$k$ items into chunks of size $\frac{n}{\sqrt{k}}$. We show that after the first round, in a typical chunk, constant fraction of items are compared to any items in the same chunk fewer than $0.1\log(k)$ times. We can then show that, given the first round comparison results, there are $\Omega(n\sqrt{k})$ pairs of items who could be a consecutively ranked pair with not small chance. As the algorithm has $O(n \sqrt{k} \log(k))$ (with a small enough constant factor) comparisons, we can conclude the algorithm could miss to compare many consecutively ranked pairs $c\cdot \log(k)$ times. The rest of the argument is similar to the 1-round lower bound. For details, see Lemma \ref{lem:2stopk_noisy_lb}.

\section{Conclusion and Open Problems}
In this paper, we characterize the optimal trade-off between the sample complexity and the round complexity of sorted top-$k$ in both the noiseless case and the noisy case. For a fixed number of rounds, our sample complexity bound is tight up to a polylogarithmic factor.

When $r =1$ or $2$, we can make our sample complexity bound of sorted top-$k$ tight up to a constant factor. We extend these results to top-$k$ and sorting. These bounds also allow us to study the blow up in the sample complexity when we transition from the noiseless case to the noisy case. Interestingly, for $r=1$ or $2$, this blow up is different in different rank aggregation problems: $\Theta(1)$ in top-$k$, $\Theta(\log(k))$ in sorted top-$k$ and $\Theta(\log(n))$ in sorting. 

There are mainly two obstacles to getting tighter bounds for top-$k$, sorting and sorted top-$k$ when we have more than 2 rounds. 
We list them as open problems here. The first one is that we don't have tight (up to a constant factor) sample complexity bounds even in the noiseless case. 
\begin{openproblem}
Get tight (up to a constant factor) sample complexity bounds for the noiseless case when $r > 2$.
\end{openproblem}
In particular, the first step is to consider 3-round top-$k$ in the noiseless case. \cite{BMW16} shows its sample complexity is $O(n \cdot polylog(n))$. \cite{BollobasB90} shows that no 3-round algorithm with $\Theta(n)$ comparisons can find top-$k$ correctly with probability $1-o(1)$. If we only want to succeed with constant probability (for example $2/3$), the best lower bound is the trivial one: $\Omega(n)$.

Once we have a good understanding of the noiseless case, we can start to think about the noisy case for $r > 2$.
\begin{openproblem}
Extend our techniques for $r=1$ or $2$ in the noisy case to the case when we have more than 2 rounds. 
\end{openproblem}
In the noisy case, our 2-round bounds are very different and more complicated compared to 1-round bounds. Even if we have tight bounds in the noiseless case, getting tight bounds for more than 2 rounds could be more difficult and might require new techniques.

\acks{We would like to thank Claire Mathieu for earlier discussions of this problem.}

\bibliography{references}
\appendix
\section{Sorted Top-$k$ in the Noiseless Case}
\label{sec:stopk}

In this section, we show upper and lower bounds on the sample complexity for solving sorted top-$k$ in the noiseless case. 

First of all, it's easy to observe that the sample complexity for solving sorted top-$k$ in 1 round is $\Theta(n^2)$. For the upper bound, we just need to compare all pairs (there are $\binom{n}{2}$ of them). For the lower bound, first observe that we can wlog assume the algorithm is deterministic. Then if the algorithm uses fewer than $\binom{n}{2}$ comparisons, it misses the comparison between the best item and the second best item with positive probability and therefore the algorithm cannot even guarantee to solve top-1.  

For more than 1 round, we show algorithms in Section \ref{sec:stopk_ub} and lower bounds in Section \ref{sec:stopk_lb}.

\subsection{Algorithms}
\label{sec:stopk_ub}
Our algorithmic results are stated in Corollary \ref{cor:stopk_2r} (for 2 rounds) and Corollary \ref{cor:stopk_3r} (for $\geq 3$ rounds). They are based on two sub-routines: Algorithm \ref{alg:rsorted1} and Algorithm \ref{alg:rsorted2}. Both of them use the sorting algorithm in \cite{AlonAV86} as a blackbox (Theorem \ref{thm:AAV86}). 

Algorithm \ref{alg:rsorted1} is used when $k$ is large ($k > n^{(2r-2)/(2r-1)}$). In the first round, we pick a random set of size $\alpha = k^{(r-1)/r} n^{(2-r)/r}$ (call them ``pivot items'') and partition the entire set into $\alpha+1$ chunks by comparing all items to the pivot items. In the remaining $r-1$ rounds, we use the sorting algorithm in \cite{AlonAV86} for each chunk that has top-$k$ items. We prove Algorithm \ref{alg:rsorted1} works in Lemma \ref{lem:rsorted1}.

Algorithm \ref{alg:rsorted2} is used when $k$ is small. Compared with Algorithm \ref{alg:rsorted1}, we pick the pivot items more carefully in Algorithm \ref{alg:rsorted2}. In the first round, we extend the result of \cite{BMW16} (stated in Theorem \ref{thm:bmw16} and Corollary \ref{cor:bmw16}) to find pivot items. In the second round, we partition the entire set into chunks by comparing all items to the pivot items. In the remaining $r-2$ rounds, we use the sorting algorithm in \cite{AlonAV86} for each chunk that has top-$k$ items. We prove Algorithm \ref{alg:rsorted2} works in Lemma \ref{lem:rsorted2}.

We first provide the final statements of our algorithmic results for sorted top-$k$ in the noiseless case:
\begin{corollary}
\label{cor:stopk_2r}
There exists a $2$-round algorithm solves sorted top-$k$ with $O(n\sqrt{k}+ n^{4/3})$ comparisons in expectation.
\end{corollary}

\begin{proof}
There are two cases:
\begin{itemize}
\item When $k > n^{2/3}$, run Algorithm \ref{alg:rsorted1} to find the sorted top-$k$. This takes $O(n\sqrt{k}) = O(n\sqrt{k}+ n^{4/3})$ comparisons in expectation.
\item When $k \leq n^{2/3}$, run Algorithm \ref{alg:rsorted1} to find the sorted top-$n^{2/3}$ and then output sorted top-$k$. This takes $O(n^{4/3}) = O(n\sqrt{k}+ n^{4/3})$ comparisons in expectation.
\end{itemize}
\end{proof}

\begin{corollary}
\label{cor:stopk_3r}
For $r \geq 3$, there exists an $r$-round algorithm solves sorted top-$k$ with $\tilde{O}(n^{2/r}k^{(r-1)/r}+ n))$ comparisons in expectation.
\end{corollary}

\begin{proof}
There are three cases:
\begin{itemize}
\item When $k > n^{(2r-2)/(2r-1)}$, run Algorithm \ref{alg:rsorted1} to find the sorted top-$k$. This takes $O(n^{2/r}k^{(r-1)/r}) = \tilde{O}(n^{2/r}k^{(r-1)/r}+ n))$ comparisons in expectation.
\item When $10 n^{(r-2)/(r-1)}\leq k \leq n^{(2r-2)/(2r-1)}$: run Algorithm \ref{alg:rsorted2} to find the sorted top-$k$. This takes $\tilde{O}(n^{2/r}k^{(r-1)/r}) = \tilde{O}(n^{2/r}k^{(r-1)/r}+ n))$ comparisons in expectation.

\item When $k<10 n^{(r-2)/(r-1)}$, run Algorithm \ref{alg:rsorted2} to find the sorted top-$10 n^{(r-2)/(r-1)}$ and then output sorted top-$k$. This takes $\tilde{O}(n)= \tilde{O}(n^{2/r}k^{(r-1)/r}+ n))$ comparisons in expectation.
\end{itemize}
\end{proof}

Now we start to show two sub-routines: Algorithm \ref{alg:rsorted1} and Algorithm \ref{alg:rsorted2}.

\begin{theorem}[\cite{AlonAV86}]
\label{thm:AAV86}
For any fixed $r > 0$, there exists an $r$-round algorithm which sorts $n$ items with $O(n^{1+1/r})$ comparisons in expectation.
\end{theorem}

\begin{algorithm}[H]
    \caption{$r$-round noiseless algorithm for sorted top-$k$ when $k>n^{(2r-2)/(2r-1)}$ and $r\geq 2$}
    	\label{alg:rsorted1}
    \begin{algorithmic}[1]
	\STATE Let $\alpha = k^{(r-1)/r} n^{(2-r)/r}$. 
	\STATE Round 1: Pick a set $S$ of $\alpha$ random items (with repetition). Compare each item in $N$ to each items in $S$. 
	\STATE Round 2 to Round $r$: Let items in $S$ have ranks $s_1\leq s_2 \leq \cdots \leq s_{\alpha}$. For notation convenience, define $s_{\alpha+1} = n+1$ and $s_0 = 0$. Let $l$ be the smallest number such that $s_l \geq k$. Define $N_i$ to be set of items that are worse than $s_{i-1}$ and better than $s_i$ for $i=1,...,l$. Use the algorithm in Theorem \ref{thm:AAV86} to sort each set $N_i$ in $r-1$ rounds in parallel.  Now we have the sorted top-$s_l$, just output the sorted top-$k$.
     \end{algorithmic}
\end{algorithm}

\begin{lemma}
\label{lem:rsorted1}
For $k>n^{(2r-2)/(2r-1)}$ and $r \geq 2$, Algorithm \ref{alg:rsorted1} is always correct and uses $O(n^{2/r}k^{(r-1)/r})$ comparisons in expectation.
\end{lemma}

\begin{proof}
The correctness of the algorithm is easy to check. To prove the lemma, it suffices to bound the expected number of comparisons used by the algorithm. In the first round, the algorithm uses $\alpha n = n^{2/r}k^{(r-1)/r}$ comparisons. From round 2 to  round $r$, by Theorem \ref{thm:AAV86}, the algorithm uses $\E[\sum_{i=1}^l |N_i|^{1+1/(r-1)} ]$ comparisons in expectation. It suffices to prove that $\E[\sum_{i=1}^l |N_i|^{1+1/(r-1)} ] = O(n^{2/r}k^{(r-1)/r})$. Notice that $l$ and $N_i$'s are random variables depending on the randomness of the algorithm. 

For $i < s_l$, define $g(i)$ such that $i$ is in $N_{g(i)}$. We have
\begin{align*}
\E\left[\sum_{i=1}^l |N_i|^{1+1/(r-1)} \right] &= \E\left[|N_l|^{1+1/(r-1)} + \sum_{i=1}^{l-1} |N_i|^{1+1/(r-1)} \right]\\
 &\leq \E\left[|N_{g(k)}|^{1+1/(r-1)}+\sum_{i=1}^k |N_{g(i)}|^{1/(r-1)}  \right].
\end{align*}
For each $i \leq k$, we will upper bound $\E[|N_{g(i)}|^{\beta}]$ for $0 < \beta \leq 2$. We start by considering $N_{g(i)} \cap \{j |j \leq i\}$ (the set of items that have ranks no worse than item $i$ and are put into the same partitioned set as item $i$). The size of this set is exactly $i-s_{g(i)-1} + 1$. For $1 \leq \beta \leq 2$,
\[
\E[|N_{g(i)} \cap \{j|j\leq i\}|^{\beta}] \leq \sum_{j=1}^n (j^{\beta} - (j-1)^{\beta})(1-j/n)^{\alpha} \leq \int_{0}^n \beta x^{\beta-1}(1-x/n)^{\alpha} dx.
\]
When $\beta = 1$, we have
\[
\int_{0}^n x^{\beta-1}(1-x/n)^{\alpha} dx = \frac{n}{\alpha + 1}< \frac{n}{\alpha}.
\]
When $\beta = 2$, we have
\[
\int_{0}^n x^{\beta-1}(1-x/n)^{\alpha} dx =\frac{n^2}{\alpha + 1} - \frac{n^2}{\alpha + 2} =  \frac{n^2}{(\alpha + 2)(\alpha + 1)}< \left(\frac{n}{\alpha}\right)^2.
\]
When $1 < \beta < 2$, by concavity of $x^{\beta -1}$ for $x>0$, we have
\[
\int_{0}^n x^{\beta-1}(1-x/n)^{\alpha} dx = \left(\int_{0}^n (1-x/n)^{\alpha} dx \right)\left(\frac{\int_{0}^n x(1-x/n)^{\alpha} dx}{\int_{0}^n (1-x/n)^{\alpha} dx }\right)^{\beta-1} < \left(\frac{n}{\alpha}\right) ^{\beta}.
\]
So  for $1 \leq \beta \leq 2$, we have
\[
\E[|N_{g(i)} \cap \{j|j\leq i\}|^{\beta}] \leq \beta \left(\frac{n}{\alpha}\right) ^{\beta}.
\]
For $\beta < 1$, by the concavity of $x^\beta$ we have
\[
\E[|N_{g(i)} \cap \{j|j\leq i\}|^{\beta}] \leq \E[|N_{g(i)} \cap \{j|j\leq i\}|]^{\beta} \leq \left(\frac{n}{\alpha}\right) ^{\beta}.
\]
By symmetry, we can also get the same upper bound on $\E[ |N_{g(i)} \cap \{j|j\geq i\}|^{\beta}]$. 
Therefore, for $0 < \beta \leq 2$:
\begin{align*}
\E[|N_{g(i)}|^{\beta}] &\leq  \E[(|N_{g(i)} \cap \{j|j\leq i\}| + |N_{g(i)} \cap \{j|j\geq i\}|)^{\beta}] \\
&\leq \E[(2|N_{g(i)} \cap \{j|j\leq i\}|)^\beta] + \E[ (2|N_{g(i)} \cap \{j|j\geq i\}|)^{\beta}] \\
&\leq O\left(\left(\frac{n}{\alpha}\right)^\beta\right).
\end{align*}
Notice that $\frac{n}{\alpha} = \frac{n^{(2r-2)/r}}{k^{(r-1)/r}}\leq k$. To sum up, we get
\[
\E\left[\sum_{i=1}^l |N_i|^{1+1/(r-1)} \right] = O\left( \left(\frac{n}{\alpha}\right)^{1+1/(r-1)} + k\cdot \left(\frac{n}{\alpha}\right)^{1/(r-1)}\right) =O(n^{2/r}k^{(r-1)/r}). 
\]
\end{proof}

\begin{algorithm}[H]
    \caption{$r$-round noiseless algorithm for sorted top-$k$ when $n^{(2r-2)/(2r-1)} \geq k\geq 10 n^{(r-2)/(r-1)}$ and $r >2$}
    	\label{alg:rsorted2}
    \begin{algorithmic}[1]
	\STATE Let $\alpha = k^{(r-1)/r} n^{(2-r)/r}$. 
	\STATE Round 1: Pick a set $S$ of $\alpha \ln(n)$ random items (with repetition). Compare each item in $N$ to each item in $S$.
	\STATE Round 1(run in parallel with the previous step): Run the algorithm of Corollary \ref{cor:bmw16} to get items $p_1,...,p_{\alpha^2+1}$ such that with probability at least $1-1/n$, $\forall i \in [\alpha^2+1], p_i$ ranks in $[\frac{i\cdot k}{\alpha^2} - \frac{1}{3} \cdot \sqrt{\frac{n}{\alpha}}, \frac{i\cdot k}{\alpha^2} +\frac{1}{3} \cdot \sqrt{\frac{n}{\alpha}}]$. 
	\STATE Round 2: Let $s$ be an item in $S$ whose rank is in $[2n/\alpha, 3n/\alpha]$. If such $s$ does not exist, declare FAIL and proceed to round 3. Let $N'$ be the set of items that are better than $s$. Compare items in $N'$ to $p_1,...,p_{\alpha^2+1}$. If $\exists i \in [\alpha^2+1]$ such that $p_i$ beats more than $\frac{i\cdot k}{\alpha^2} +\frac{1}{3} \cdot \sqrt{\frac{n}{\alpha}}$ or less than $\frac{i\cdot k}{\alpha^2} -\frac{1}{3} \cdot \sqrt{\frac{n}{\alpha}}$ items in $N'$, declare FAIL.
	\STATE Round 3-$r$: If the algorithm declares FAIL in round 2, compare all pairs of items and output the sorted top-$k$. Otherwise assume $p_1,...,p_{\alpha^2}$ partition items worse than $p_{\alpha^2+1}$ into sets $N_1, ... ,N_{\alpha^2}$. Use the algorithm in Theorem \ref{thm:AAV86} to sort each set $N_i$ in $r-2$ rounds in parallel. Now we have the sorted top-$p_{\alpha^2+1}$, output the sorted top-$k$. 
     \end{algorithmic}
\end{algorithm}

Algorithm 5 in Appendix C.1.1 of \cite{BMW16} can be easily extended to show the following theorem. In that algorithm, for their purpose, only $p_i$ for $i=n/2$ is explicitly computed after one round (denoted as $x$ in their pseudocode). However, it is not hard to see that $p_i$'s for all $i\in[n]$ can be computed in the same way using the same set of comparisons. The only change is that the failure probability is multiplied by a factor of $n$ because of union bound.

\begin{theorem}[\cite{BMW16}]
\label{thm:bmw16}
There exists an 1-round algorithm with $O(n)$ comparisons which outputs a list of item $p_i$'s for all $i \in [n]$ such that with probability at least $1-1/n$, $\forall i \in n$, $p_i$'s rank is at most $C(n) \cdot \sqrt{n}$ away from $i$ for some $C(n) = polylog(n)$.  
\end{theorem}

Using Theorem \ref{thm:bmw16}, we can get the following corollary. 
\begin{corollary}
\label{cor:bmw16}
For any $\alpha > 1$, there exists an 1-round algorithm with $\tilde{O}(\alpha n)$ comparisons which outputs a list of $p_i$'s for all $i \in [n]$ such that with probability at least $1-1/n$, $\forall i \in n$, $p_i$'s rank is at most $\frac{1}{3} \cdot \sqrt{\frac{n}{\alpha}}$ away from $i$. 
\end{corollary}

\begin{proof}
Set $\beta$ to be some value larger than $9\alpha \cdot (C(\beta n))^2 $($C(n)$ is the one in the statement of Theorem \ref{thm:bmw16}). It suffices to pick some $\beta = polylog(n)$. For each item, create $\beta$ copies. Run the 1-round algorithm of Theorem \ref{thm:bmw16} on these copies ($\beta n$ items). If the algorithm compares the copies of different items, we make an actual comparison between these two items. If the algorithm compares the copies of the same item, we just decide the comparison result based on some arbitrarily fixed order between the copies of the same item. Suppose the 1-round algorithm outputs $q_1,...,q_{\beta n}$. We output $p_i = q_{\beta \cdot i}$. It's easy to check that with probability at least $1-1/n$, $\forall i \in n$, the difference between $p_i$'s rank and $i$ is at most
\[
\frac{ C(\beta n) \sqrt{\beta n}  }{\beta} \leq \frac{1}{3} \cdot \sqrt{\frac{n}{\alpha}}. 
\]
\end{proof}

\begin{lemma}
\label{lem:rsorted2}
For $n^{(r-2)/(r-1)}<k\leq n^{(2r-2)/(2r-1)}$ and $r \geq 3$, Algorithm \ref{alg:rsorted2} is always correct and uses $\tilde{O}(n^{2/r}k^{(r-1)/r})$ comparisons in expectation. 
\end{lemma}

\begin{proof}
We first prove the correctness of Algorithm \ref{alg:rsorted2}. If the algorithm declares FAIL in round 2, then all the pairs of items get compared. The algorithm definitely outputs the sorted top-$k$ correctly. Now we consider the case when the algorithm does not declare FAIL. First the set $N'$ will have all the items in top-$s$ for some $s \in [2n/\alpha, 3n/\alpha]$. Since $n/\alpha \geq k$,  $N'$ contains all the items in top-$k$. We also have 
\[
\frac{k}{\alpha^2} \geq \sqrt{\frac{n}{\alpha}} \Leftrightarrow n^{2r-6} \geq k^{r-3}. 
\]
Since $r \geq 3$, we have $\frac{k}{\alpha^2} \geq \sqrt{\frac{n}{\alpha}}$. For all $i \in [\alpha^2+1]$, we have $\frac{i\cdot k}{\alpha^2} +\frac{1}{3} \cdot \sqrt{\frac{n}{\alpha}} < k(1+\frac{1}{2} + \frac{1}{3}) < s$. Therefore if the algorithm does not declare FAIL, we have $\forall i \in [\alpha^2+1], p_i$ ranks in $[\frac{i\cdot k}{\alpha^2} - \frac{1}{3} \cdot \sqrt{\frac{n}{\alpha}}, \frac{i\cdot k}{\alpha^2} +\frac{1}{3} \cdot \sqrt{\frac{n}{\alpha}}]$. Therefore we know that all these $p_i$'s are in $N'$. We also have $p_{\alpha^2+1} \geq \frac{(\alpha^2+1)\cdot k}{\alpha^2} - \frac{1}{3} \cdot \sqrt{\frac{n}{\alpha}} >  \frac{(\alpha^2+1)\cdot k}{\alpha^2} - \frac{1}{3} \cdot \frac{k}{\alpha^2} \geq k$. This means the top-$k$ is inside the top-$p_{\alpha^2+1}$. Since the algorithm gets sorted top-$p_{\alpha^2+1}$, it can output sorted top-$k$ correctly. 

Now we want to bound the expected number of comparisons of the algorithm. In the first round, both steps use $\tilde{O}(\alpha n) = \tilde{O}(n^{2/r}k^{(r-1)/r})$ comparisons. In the second round, the algorithm uses at most $(\alpha^2+1) \cdot \frac{ 3n}{\alpha} = O(\alpha n) = O(n^{2/r}k^{(r-1)/r})$ comparisons. In the third round, if the algorithm does not declare FAIL, each $N_i$ will have size at most $\frac{2k}{\alpha^2}$. By Theorem \ref{thm:AAV86}, the algorithm will use at most 
\[
\left(\frac{2k}{\alpha^2}\right)^{1+1/(r-2)} \cdot \alpha^2 = O(k^{1+1/(r-2)} \cdot \alpha^{-2/(r-2)}) = O(n^{2/r}k^{(r-1)/r})
\] comparisons in expectation. If the algorithm declares FAIL, the algorithm will use $O(n^2)$ comparisons. We know the probability of FAIL is at most $1/n+ (1-\frac{1}{\alpha})^{\alpha \ln(n)} < 2/n$. Therefore in expectation, FAIL will cause at most $O(n) = O(n^{2/r}k^{(r-1)/r})$ comparisons. 
\\
\end{proof}


\subsection{Lower Bounds}
\label{sec:stopk_lb}
We prove lower bounds in Lemma \ref{lem:stopk_lb3} (for $\geq 3$ rounds) and Lemma \ref{lem:stopk_lb2} (for 2 rounds). The main idea is to show that there are not many comparisons between top-$k$ items in the first 1 or 2 rounds and then reduce from sorting $k$ items for the remaining $r-1$ or $r-2$ rounds. Similarly as our algorithms, we use the lower bound for sorting (Theorem 2.1 of \cite{AlonA88b}) as a blackbox. We state their theorem (Theorem \ref{thm:AA88}) in the format which is enough for our proof.

\begin{theorem}[\cite{AlonA88b}]
\label{thm:AA88}
Suppose there are $2$ disjoint sets of item, denoted $Z$ and $Y$. $|Y| = y$. The rank of each item in $Z$ is known. The set of $y$ ranks in $Y$ is known but all the $y!$ orders of the items of the set are equally likely. Suppose in the first round the algorithm already makes $e$ comparisons between 2 items in $Y$ and $e'$ comparisons between an item in $Y$ and an item in $Z$. Let $f = e + e'/2$. The expected number of comparisons to sort all the items in $r$ more rounds is at least
\[
r \cdot \left( \frac{y^{1+1/r}}{c\cdot g_r(y,f)} - y\right).
\]
for some constant $c$. Here $g_r(y,f)$ is defined as 
\[
g_r(y,f) = \left\{
\begin{array}{rcl}
 1,  ~~~~~~~~&      & {f=0},\\
 \left(\frac{c}{4}\right)^{1/r} ,~~   &      & {0\leq \frac{f}{y}\leq \frac{1}{4}},\\
 \left(\frac{cf}{y}\right)^{1/r},~      &      & { \frac{f}{y}\geq \frac{1}{4}}.\\
\end{array} \right.
\]
\end{theorem}

\begin{lemma}
\label{lem:stopk_lb3}
For $r \geq 3$, $r$-round algorithm needs $\Omega(n^{2/r} k^{(r-1)/r} +n)$ comparisons in expectation to solve sorted top-$k$.
\end{lemma}

\begin{proof}
Wlog we assume the algorithm gives each item a label in $[n]$. As the algorithm has no information about the ordering, the list of these labels (the label of the rank-1 item, the label of the rank-2 item,...) is a uniformly random permutation of $[n]$. After that, notice that an $r$-round randomized algorithm is just a distribution over $r$-round deterministic algorithms. To prove a lower bound on the expected number of comparisons, it suffices to only consider deterministic algorithms. And now the randomness only comes from how items are labeled. There are three cases to consider depending on the values of $n$ and $k$:

\textbf{Case 1: $k \geq n^{(2r-2)/(2r-1)}$.} If in the first round, the algorithm is using at least $n^{2/r} k^{(r-1)/r}$ comparisons, then we already have the lower bound. Now assume that in the first round, the algorithm is using fewer than $n^{2/r} k^{(r-1)/r}$ comparisons. In expectation, only $O(\frac{k^2}{n^2})$ fraction of comparisons in the first round are between two items in top-$k$ (we call them useful comparisons). Therefore in expectation, the algorithm has fewer than $O(n^{2/r} k^{(r-1)/r} \cdot (\frac{k}{n})^2)$ useful comparisons. By giving some extra useful comparisons to the algorithm's first round, we can make sure that the algorithm always have at least $k/4$ useful comparisons in the first round and the expected number of useful comparisons in the first round is $O(n^{2/r} k^{(r-1)/r} \cdot (\frac{k}{n})^2 + k ) = O(n^{2/r} k^{(r-1)/r} \cdot (\frac{k}{n})^2)$.

Now consider the case that we tell the algorithm the set of top-$k$ items after the first round and the algorithm just need to sort the set of $k$ items in the remaining $r-1$ rounds. By Theorem \ref{thm:AA88} and the convexity of $f(x) = 1/x$ for $x>0$, the number of comparisons we need to use in the last $r-1$ rounds is at least 
\[
\Omega\left(\frac{k^{1+1/(r-1)}}{(n^{2/r} k^{(r-1)/r} \cdot (\frac{k}{n})^2\cdot \frac{1}{k})^{1/(r-1)} }\right) = \Omega(n^{2/r} k^{(r-1)/r}).
\]

\textbf{Case 2: $n^{(r-2)/(r-1)} \leq k < n^{(2r-2)/(2r-1)}$.} Set $w = \frac{n^{(2r-2)/r}}{k^{(r-1)/r}}$. If in the first round, the algorithm is using at least $\frac{1}{8}\cdot n^{2/r} k^{(r-1)/r}$ comparisons, then we already have the lower bound. Now assume that in the first round, the algorithm is using less than $\frac{1}{8} \cdot n^{2/r} k^{(r-1)/r}$ comparisons. In this case the expected number of comparisons between 2 items in top-$w$ is at most 
\[
\frac{1}{8} \cdot n^{2/r} k^{(r-1)/r} \cdot \frac{\binom{w}{2}}{\binom{n}{2}} \leq \frac{1}{8} \cdot n^{2/r} k^{(r-1)/r} \cdot \frac{w^2}{n^2} \leq \frac{1}{8} \cdot w.
\]
Define $W$ to be the set of items in top-$w$ that have no comparisons with any other items in top-$w$ in the first round. Define $K$ to be the set of items in top-$k$ that have no comparisons with any other items in top-$w$ in the first round. We have $\E[|W|]] \geq w -  \frac{1}{8} \cdot w \cdot 2 \geq 3w/4$ and $\E[|K|] = \E[|W|] \cdot \frac{k}{w} \geq 3k/4$. Then by Markov inequality, we have that with probability $1-1/3-1/3 = 1/3$, both $|W| \geq w/2$ and $|K| \geq k/2$. 

Fix the comparison results of the first round. Now we will focus on the situation when $|W| \geq w/2$ and $|K| \geq k/2$. It suffices to show that in this case, the algorithm needs $\Omega(n^{2/r} k^{(r-1)/r} )$ comparisons in expectation to solve sorted top-$k$ in $r-1$ rounds. We remove some items from $K$ and $W$ such that $K \subseteq W$, $|K| = k/2$ and $|W| = w/2$. We also tell the algorithm the rank of each item in $[n] \backslash W$. Now consider the next round of the algorithm (the second round). If the algorithm uses at least $n^{2/r} k^{(r-1)/r}$ comparisons in expectation, then we are done. Now assume the algorithm uses fewer than $n^{2/r} k^{(r-1)/r}$ comparisons in expectation in the second round. Call comparisons involving at least one item in $K$ as useful comparisons. Since the algorithm has no information about which items in $W$ are in $K$, the algorithm has at most 
$n^{2/r} k^{(r-1)/r} \cdot \frac{2k}{w} = \frac{2k^{(3r-2)/r}}{n^{(2r-4)/r}}$ useful comparisons in expectation. By giving some extra useful comparisons to the algorithm's second round, we can make sure that the algorithm always have at least $k/2$ useful comparisons in the second round and the expected number of useful comparisons in the second round is $O(\frac{k^{(3r-2)/r}}{n^{(2r-4)/r}} + k ) = O(\frac{k^{(3r-2)/r}}{n^{(2r-4)/r}})$.

By Theorem \ref{thm:AA88} and the convexity of $f(x) = 1/x$ for $x>0$, the number of comparisons we need to use in the last $r-2$ rounds is at least 
\[
\Omega\left(\frac{k^{1+1/(r-2)}}{\left(\frac{k^{(3r-2)/r}}{n^{(2r-4)/r}} \cdot \frac{1}{k}\right)^{1/(r-2)}} \right)= \Omega(n^{2/r} k^{(r-1)/r}).
\]

\textbf{Case 3: $k < n^{(r-2)/(r-1)}$.} In this case, we have $n^{2/r} k^{(r-1)/r} < n$. We know that just to find the set of top-$k$ without round constraint we need at least $\Omega(n)$ comparisons. This simply gives the lower bound. 
\end{proof}

\begin{lemma}
\label{lem:stopk_lb2}
 Any $2$-round algorithm needs $\Omega(n \sqrt{k}+n^{4/3})$ comparisons in expectation to solve sorted top-$k$.
\end{lemma}

\begin{proof}
Similarly as Lemma \ref{lem:stopk_lb3}, it suffices to prove the lower bound only for deterministic algorithms. There are two cases:

\textbf{Case 1: $k \geq n^{2/3}$.} Notice that this case is the same as Case 1 in the proof of Lemma \ref{lem:stopk_lb3} and that proof also works when $r=2$. So from that we get a lower bound $\Omega(n \sqrt{k}) = \Omega(n \sqrt{k}+ n^{4/3})$. 

\textbf{Case 2: $k < n^{2/3}$.}  In this case we have $n^{4/3} > n \sqrt{k}$. Notice that if an algorithm solves sorted top-$k$, it also finds top-1. \cite{AlonA88b} claims in the concluding remark that a 2-round algorithm which finds top-1 needs $\Omega(n^{4/3})=\Omega(n \sqrt{k}+ n^{4/3})$ comparisons in expectation. This directly implies the lemma in this case. 
\end{proof}

\section{Top-$k$ in the Noisy Case}
\label{sec:topk_n}

In this section, we show algorithms for top-$k$ in the noisy case. In particular, we show an 1-round top-$k$ algorithm with sample complexity $O(n^2)$ in Section \ref{sec:topk1} and a 2-round top-$k$ algorithm with sample complexity $O(n^{4/3})$ in Section \ref{sec:topk2}.

\subsection{1-Round Top-$k$ Algorithm in the Noisy Case}
\label{sec:topk1}

In this sub-section, we show an 1-round top-$k$ algorithm with sample complexity $O(n^2)$ (Lemma \ref{lem:1topk}).
 
\begin{algorithm}[H]
    \caption{$1$-round algorithm for top-$k$ with noisy comparisons}
    	\label{alg:1topk}
    \begin{algorithmic}[1]
    	\STATE Set $c = 18 \cdot 48$, $l_0 = n$ and $l_i =\max\left( \underbrace{\log...\log}_{i \text{ times}} (n), 1\right)$ for $i = 1,...,\log^*(n)$ . 
	\STATE For $i=1,...,\log^*(n)$, pick a set $S_i$ of $n/l_i^2$ random items.
    	\STATE \textbf{Round 1:}
	\begin{ALC@g}
    	\FOR {$i = 1$ to $\log^*(n)$}
		\STATE  For each item in $N$ and each item in $S_i$, compare them $c \cdot l_i$ times.
	\ENDFOR
	\end{ALC@g}
	\STATE \textbf{Output procedure:}
	\begin{ALC@g}
	\STATE Set $N_1 = [n]$, $T_0 = \emptyset$ and $k_1 = k$. 
    	\FOR {$i = 1$ to $\log^*(n)$}
		\STATE For each item $j \in S_i$ and item $j' \in N_i$, if $j$ wins the majority of comparisons among $c \cdot l_i$ comparisons between $j$ and $j'$ in the $i$-th iteration, say $j$ beats $j'$ in the $i$-th iteration. Otherwise say $j'$ beats $j$ in the $i$-th iteration.
		\STATE For each item $j  \in S_i \cap N_i$, define $r_i(j)$ as the number of items in $N_i$ which are not beaten by item $j$ in the $i$-th iteration. 
		\STATE Let $a_i = \argmin_{j \in S_i \cap N_i, r_i(j) \leq k_i} |r_i(j) - k_i|$ and $A_i$ be the set of items in $N_i$ which are not beaten by item $a_i$ in the $i$-th iteration. If $a_i$ does not exist, set $A_i = \emptyset$. 
		\STATE Let $b_i = \argmin_{j \in S_i \cap N_i, r_i(j) > k_i} |r_i(j) - k_i|$ and $B_i$ be the set of items in $N_i$ which don't beat item $b_i$ in the $i$-th iteration. If $b_i$ does not exist, set $B_i = \emptyset$. 
		\STATE $T_i = T_{i-1} \cup A_i$. 
		\STATE $N_{i+1} = N_i \backslash (A_i \cup B_i)$
		\STATE $k_{i+1} = k_i  - |A_i|$. 
	\ENDFOR
	\STATE Output $T_{\log^*(n)}$ as the top-$k$ set. 
	\end{ALC@g}
     \end{algorithmic}
\end{algorithm}

\begin{lemma}
\label{lem:1topk}
Algorithm \ref{alg:1topk} uses $O(n^2)$ comparisons and outputs top-$k$ correctly with probability at least $2/3$. 
\end{lemma}

\begin{proof}
The number of comparisons used is at most
\begin{align*}
&cn^2+  \sum_{i=1}^{\log^*(n)-1} n \cdot \frac{n}{l_i^2} \cdot c \cdot l_i \\
&= cn^2 \cdot \left(1+ \sum_{i=1}^{\log^*(n)-1}\frac{1}{l_i} \right)\\
&\leq cn^2 \cdot  \left(1+\sum_{i=1}^{\log^*(n)-1} \frac{1}{2^{i-1}} \right)\\
&= O(n^2).
\end{align*}
Now we are going to show Algorithm \ref{alg:1topk} is correct with probability at least $2/3$. Consider event $W$ as the intersection of the following events:
\begin{itemize}
\item $W_1$: For each $i = 1,..., \log^*(n)-1$, either $k - 4l_i^3< 1$ or there exists an item $j \in S_i$ such that item $j$ ranks between $k -4l_i^3$ and $k$. 
\item $W_2$: For each $i = 1,...,\log^*(n) - 1$, either $k +4l_i^3 > n$ or there exists an item $j \in S_i$ such that item $j$ ranks between $k+1$ and $k +4l_i^3$. 
\item $W_3$: For each $i = 1,...,\log^*(n)$ and each pair of items $j$ and $ j'$ such that $j\in S_i$ and both $j$ and $j'$ rank between $\min\left(1, k-4l_{i-1}^3\right)$ and $\max\left(n, k+4l_{i-1}^3\right)$, whether $j$ beats $j'$ in the $i$-th iteration is consistent with the true ordering. 
\end{itemize}
Assuming $W$ happens, it is easy to check by induction that the following is true, for $i=1,...,\log^*(n)$:
\begin{itemize}
\item  $N_i$ is a subset of items rank between $\min\left(1, k-4l_{i-1}^3\right)$ and $\max\left(n, k+4l_{i-1}^3\right)$. Items in $N_i$ have consecutive ranks. The $k_i$-th item in $N_i$ is the $k$-th item of the entire set. 
\item $r_i(j)$ is the true rank of item $j$ in $N_i$.
\item Items in $A_i$ are all in top-$k$ and items in $B_i$ are all in bottom-$(n-k)$.
\end{itemize}
Finally notice that in the last iteration $S_{\log^*(n)} = [n]$ and $N_{\log^*(n)+1} = \emptyset$. Therefore all items are placed in some $A_i$ or $B_i$. Since $T_{\log^*(n)} = \bigcup_{i=1}^{\log^*(n)} A_i$, the algorithm outputs correctly. To sum up, so far we have that, $W$ implies the algorithm outputs correctly. 

Now it suffices to show that $W$ happens with probability at least $2/3$. We are going to analyze $W_1,W_2,W_3$ and then take a union bound.
\begin{itemize}
\item The probability that $W_1$ does not hold (i.e. $1-\Pr[W_1]$) is at most
\[
\sum_{i=1}^{\log^*(n)-1} (1-4l_i^3/n)^{\frac{n}{l_i^2}} \leq \sum_{i=1}^{\log^*(n)-1}e^{-4l_i} \leq 2 e^{-4l_{\log^*(n)-1}} \leq 2e^{-4}.
\]
\item Similarly as the previous bullet, $\Pr[W_2] \geq 1- 2e^{-4}$. 
\item In the $i$-th iteration, for each pair of items $j$ and $ j'$ such that $j\in S_i$ and  both $j$ and $j'$ rank between $\min\left(1, k-4l_{i-1}^3\right)$ and $\max\left(n, k+4l_{i-1}^3\right)$, there are $c \cdot l_i$ comparisons between $j$ and $j'$. By Chernoff bound, the majority of these comparisons differs from the true ordering with probability at most 
\[
\exp\left(-\frac{2}{3} c \cdot l_i \cdot \left(\frac{1}{4}\right)^2 \cdot \frac{1}{2}\right) = \exp(-c \cdot l_i/ 48). 
\]
By union bound, 
\begin{align*}
1-\Pr[W_3] &\leq \sum_{i=1}^{\log^*(n)} (8l_{i-1}^3)^2 \cdot \exp(-c \cdot l_i/ 48)\\
& \leq \sum_{i=1}^{\log^*(n)} (8l_{i-1}^3)^2 \cdot \frac{1}{2^{11} \cdot l_{i-1}^7 } \leq 2 \cdot \frac{1}{2^5 \cdot  l_{log^*(n)-1} } \leq \frac{1}{16}.
\end{align*}
\end{itemize}
To sum up, by union bound,
\[
\Pr[W] \geq 1 - (1-\Pr[W_1]) - (1-\Pr[W_2])-(1-\Pr[W_3]) \geq 1 - \frac{4}{e^{4}} - \frac{1}{16} > 2/3.
\]
\end{proof}

\subsection{2-Round Top-$k$ Algorithm in the Noisy Case}
In this sub-section, we show a 2-round top-$k$ algorithm with sample complexity $O(n^{4/3})$ (Corollary \ref{cor:2topk}).

For convenience we will assume $\min(k,n-k) \geq 40 \cdot \log(n) \cdot n^{2/3}$. When $\min(k,n-k)$ is small, we can add $\Theta(\log(n) \cdot n^{2/3})$ dummy items as the top and bottom ones and then the problem reduces to the case when $\min(k,n-k)$ is large. This procedure only blows up the sample complexity by a constant factor. 

\label{sec:topk2}
\begin{algorithm}[H]
    \caption{$2$-round algorithm for top-$k$ with noisy comparisons}
    	\label{alg:2topk}
    \begin{algorithmic}[1]
    	\STATE Halt the algorithm whenever it uses more than $c_0 \cdot n^{4/3}$ comparisons. $c_0$ is a constant specified in the proof. 
    	\STATE \textbf{Round 1:} Run Algorithm \ref{alg:2topk1}.
    	\STATE \textbf{Round 2:} Run Algorithm \ref{alg:2topk2}.
		\STATE Output $T_{\log^*(n)}$ as the top-$k$ set. 

     \end{algorithmic}
\end{algorithm}

\begin{algorithm}[H]
    \caption{First round of Algorithm \ref{alg:2topk}}
    	\label{alg:2topk1}
    \begin{algorithmic}[1]
    	\STATE Comparisons:
		\begin{ALC@g}
    	\STATE (a) Pick a set $S$ of $n^{1/3}$ random items. For each item in $N$ and each item in $S$, compare them $c_1 = 72\cdot 32$ times. 
    	\STATE (b) For each pair of items in $S$, compare them $100\log(n)$ times. 
    	\end{ALC@g}
	\STATE Sort items in $S$ according to comparisons in (b): label items in $S$ as $s_1,...,s_{n^{1/3}}$ such that for each $i<j$, $s_i$ wins the majority of comparisons between $s_i$ and $s_j$ in (b). If such labeling does not exist, label them arbitrarily. 
	\STATE For each $i \in N$ and $s_j \in S$, set $X_{i,j} = 1$ if $j$ wins the majority of the comparisons in step (a). Otherwise set $X_{i,j} = -1$. Let $p(i) =\arg \max_J \sum_{j=1}^J X_{i,j}$ for each $i \not\in S$. For each $s_j \in S$, set $p(s_j) = j$. Let $P_j = \{i | p(i) = j\}$. 
	\STATE Let $m$ be the minimum $J$ such that $\sum_{j=1}^J |P_j| \geq k$. 
     \end{algorithmic}
\end{algorithm}

\begin{algorithm}[H]
    \caption{Second round of Algorithm \ref{alg:2topk}}
    	\label{alg:2topk2}
    \begin{algorithmic}[1]
\STATE Set  $l_i =\max\left( \underbrace{\log...\log}_{i \text{ times}} (n), 1\right)$ for  $i = 1,...,\log^*(n)$ and $l_0 = n$. Pick a set $S_i$ of $ |\bigcup_{j=m-l_i}^{m+l_i} P_j|/ l_i^4$ random items in $\bigcup_{j=m-l_i}^{m+l_i} P_j$ for $i = 1,...,\log^*(n)$. 
	\STATE \textbf{Comparisons:} 
	\begin{ALC@g}
    	\FOR {$i = 1$ to $\log^*(n)$}
		\STATE  For each item in $\bigcup_{j=m-l_i}^{m+l_i} P_j$ and each item in $S_i$, compare them $c_2 \cdot l_i$ times.
	\ENDFOR
	\end{ALC@g}
	\STATE Set $N_1 = \bigcup_{j=m-l_i}^{m+l_i} P_j$, $T_0 = \bigcup_{j=0}^{m-l_1-1} P_j$ and $k_1 = k - \sum_{j=0}^{m-l_1-1} |P_j|$. 
    	\FOR {$i = 1$ to $\log^*(n)$}
		\STATE For each item $j \in S_i$ and item $j' \in N_i$, if $j$ wins the majority of comparisons among $c_2 \cdot l_i$ comparisons between $j$ and $j'$ in the $i$-th iteration, say $j$ beats $j'$ in the $i$-th iteration. Otherwise say $j'$ beats $j$ in the $i$-th iteration.
		\STATE For each item $j  \in S_i \cap N_i$, define $r_i(j) = $ number of items in $N_i$ which are not beaten by item $j$ in the $i$-th iteration. 
		\STATE Let $a_i = \argmin_{j \in S_i \cap N_i, r_i(j) \leq k_i} |r_i(j) - k_i|$ and $A_i$ be the set of items in $N_i$ which are not beaten by item $a_i$. If $a_i$ does not exist, set $A_i = \emptyset$. 
		\STATE Let $b_i = \argmin_{j \in S_i \cap N_i, r_i(j) > k_i} |r_i(j) - k_i|$ and $B_i$ be the set of items in $N_i$ which don't beat item $b_i$. If $b_i$ does not exist, set $B_i = \emptyset$. 
		\STATE $T_i = T_{i-1} \cup A_i$. 
		\STATE $N_{i+1} = N_i \backslash (A_i \cup B_i)$
		\STATE $k_{i+1} = k_i  - |A_i|$. 
	\ENDFOR
     \end{algorithmic}
\end{algorithm}

We prove the following corollary for Algorithm \ref{alg:2topk}. It's based on Lemma \ref{lem:2r-1r} (for the first round of the algorithm) and Lemma \ref{lem:2r-2r} (for the second round of the algorithm). 
\begin{corollary}
\label{cor:2topk}
Algorithm \ref{alg:2topk} uses $O(n^{4/3})$ comparisons and outputs top-$k$ correctly with probability at least $2/3$.
\end{corollary}

\begin{proof}
First of all, since the algorithm halts whenever it uses more than $c_0 \cdot n^{4/3}$ comparisons, the total number of comparisons used is always at most $O(n^{4/3})$. 

By Lemma \ref{lem:2r-1r} and Lemma \ref{lem:2r-2r}, we have
\[
\Pr[W_1 \cap W_2] \geq 1 -1/6 -1/6 \geq 2/3. 
\]
When both $W_1$ and $W_2$ hold, $W_{1,6}$, $W_{2,3}$ and $W_{2,4}$ imply that the algorithm outputs correctly. 
\end{proof}

\subsubsection{First Round of the 2-Round Algorithm}
We first analyze the first round of the algorithm. Consider $W_1$ as the intersection of the following events ($W_1$ happens when all of them happen):
\begin{itemize}
\item $W_{1,1}$: Items in $S$ are correctly sorted.
\item $W_{1,2}$: $s_m$'s rank is between $k - n^{2/3}/20$  and $k-6n^{2/3}$. 
\item $W_{1,3}$: $s_{m+1}$'s rank is between $k+n^{2/3}/20$ and $k+6n^{2/3}$. 
\item $W_{1,4}$: $\forall i = 1,...,\log^*(n)$, $|\bigcup_{j=m-l_i}^{m+l_i} P_j| \leq 100 \cdot l_i\cdot  n^{2/3}$.
\item $W_{1,5}$: $\forall i = 1,...,\log^*(n)-1$, $ \{k-6 l_i^5+1,..., k+ 6 l_i^5\} \subseteq \bigcup_{j=m-l_{i+1}}^{m+l_{i+1}} P_j $
\item $W_{1,6}$: $\bigcup_{j< m-l_1} P_j \subseteq \{1,...,k\}$ and $\bigcup_{j> m+l_1} P_j \subseteq \{k+1,...,n\}$
\end{itemize} 
We use $W_1$ to indicate the success of the first round. We show in Lemma \ref{lem:2r-1r} that $W_1$ happens with probability at least $1-1/6$.

Before proving Lemma \ref{lem:2r-1r}, we first prove Lemma \ref{lem:brw} and Lemma \ref{lem:cj} which analyze the biased random walk and are used in the proof of Lemma \ref{lem:2r-1r}.

\begin{lemma}
\label{lem:brw}
Suppose $S$ are correctly sorted. For each $J$, $s_J \leq i < s_{J+1}$ and $t$,  $\Pr[p(i) - J = t] \leq \alpha^{|t|}$ and $\Pr[|p(i) - J| \geq |t|] \leq 3\alpha^{|t|}$  where $\alpha = 1/e^{15}$. 
\end{lemma}

\begin{proof}
When $t =0$, the inequality. By symmetry, it suffices to prove the case when $t > 0$.  If $J + t > n^{1/3}$, $\Pr[p(i) - J = t]  = 0$. Now let's focus on the case when $j+t \leq n^{1/3}$. We have
\[
\Pr[p(i) -J= t] \leq \Pr\left[\sum_{j=1}^J X_{i,j} \leq \sum_{j=1}^{J+t} X_{i,j}\right] = \Pr\left[\sum_{j=J+1}^{J+t} X_{i,j} \leq 0\right].
\]
For each $j \in \{J+1, ...,J+t\}$, by Chernoff bound, we know $\Pr[X_{i,j}=1] > 1-e^{-32}$ . 

By Chernoff bound again, we have
\[
\Pr\left[\sum_{j=J+1}^{J+t} X_{i,j} \leq 0\right] \leq \exp(- \D(1/2 \| 1-e^{-22})\cdot t) = \left(\frac{1}{4(1-e^{-22})e^{-22}} \right)^{-t/2} \leq e^{-15t}.
\]
Finally we have 
\[
\Pr[|p(i) - J| \geq |t|] \leq  \frac{2\alpha^{|t|}}{1-\alpha} \leq 3\alpha^{|t|}.
\]
\end{proof}

\begin{lemma}
\label{lem:cj}
Suppose $S$ are correctly sorted. For each $J$, let $C_J = \sum_{j<J} |P_j| = |\{i | p(i) < J\}|$. For any $\beta > 0$, we have $\Pr[|C_J - s_J| > \beta \cdot n^{2/3}] \leq \frac{1}{5000\beta}$.
\end{lemma}

\begin{proof}
By symmetry,  it suffices to prove that $\Pr[C_J -s_J > \beta \cdot n^{2/3}] \leq \frac{1}{10000\beta}$. Define $\Delta_J = \{ i | i < s_J, p(i) \geq J\}$ and $\delta_J = |\Delta_J|$. We have $C_J - s_J \leq \delta_J$. We will first give an upper bound on $\E[\delta_J]$. By Lemma \ref{lem:brw} and Chu-Vandermonde identity, we have
\begin{align*}
\E[\delta_J] &\leq \sum_{i=1}^{s_{J}-1} \sum_t \Pr[ \{i+1,...,s_J-1\} \cap S| = t]  \cdot \alpha^{t+1}\\
&=  \sum_{i=1}^{s_J-1} \sum_{t=0}^{n^{1/3}-1} \frac{\binom{s_J-i-1}{t} \cdot \binom{n-s_J+i}{n^{1/3}-t-1}}{\binom{n-1}{n^{1/3}-1}} \cdot \alpha^{t+1} \\
&\leq   \sum_{t=0}^{n^{1/3}-1} \alpha^{t+1} \cdot \sum_{i=0}^{s_{J}-1} \frac{\binom{s_J-i-1}{t} \cdot \binom{n-s_J+i}{n^{1/3}-t-1}}{\binom{n-1}{n^{1/3}-1}} \\
&=  \sum_{t=0}^{n^{1/3}-1} \alpha^{t+1} \cdot \frac{\binom{n}{n^{1/3}}}{\binom{n-1}{n^{1/3}-1}}\\
&\leq \frac{\alpha n^{2/3}}{1-\alpha}. 
\end{align*}
Since random variable $\delta_J \geq 0$, by Markov inequality, we have
\[
\Pr[\delta_J \geq \beta \cdot n^{2/3} ] \leq \frac{\frac{\alpha n^{2/3}}{1-\alpha}}{\beta \cdot n^{2/3}} \leq \frac{2\alpha}{\beta} <\frac{1}{10000\beta}.
\]
\end{proof}

\begin{lemma}
\label{lem:2r-1r}
$W_1$ happens with probability at least $1-1/6$. 
\end{lemma}

\begin{proof}
The randomness in the first round of the algorithm comes from three independent components: 
\begin{enumerate}[(i)]
\item Randomness used to pick the random set $S$.
\item Randomness in the noise of comparisons in comparisons (b). These comparisons are used to sort items in $S$. 
\item Randomness in the noise of comparisons in comparisons (a). These comparisons are used to compute $p(i)$. 
\end{enumerate}

They together decide whether $W_1$ happens. We are going to analyze (i) and (ii) separately and then analyze (iii) when fixing the randomness of (i) and (ii).

For (ii), the related event in $W_1$ is $W_{1,1}$. For each pair $s_i,s_j \in S$, by Chernoff bound, the majority of comparisons between $s_i$ and $s_j$ in comparisons (b) is consistent with the true ordering with probability at least $1-1/n^2$. By union bound, $\Pr[W_{1,1}] \geq 1- 1/n$. 

For (i), consider $W_1'$ as the intersection of following events. They are all about the random set $S$. Later we will use $W_1'$ to analyze $W_1$. 
\begin{itemize}
\item $W'_{1,1}$: $S \cap \{ k -n^{2/3}/40+1 ,..., k+n^{2/3}/40\} = \emptyset$.
\item $W'_{1,2}$: $S \cap \{ k - 6 n^{2/3},...,k-n^{2/3}/40\} \neq \emptyset$.
\item $W'_{1,3}$: $S \cap \{ k + n^{2/3}/40+1,...,k+6 n^{2/3}\} \neq \emptyset$.
\item $W'_{1,4}$: $\forall i = 1,...,\log^*(n)$, $|S \cap \{k-40\cdot l_i \cdot n^{2/3}+1,..., k\}| \geq l_i$ and $|S \cap \{k+1,...,k+40\cdot l_i \cdot n^{2/3}\}| \geq l_i$. 
\end{itemize}
Now we analyze the probability $W_1'$ holds (i.e. $\Pr[W_1']$). We have 
\[
\Pr[W_{1,1}'] = \frac{\binom{n-n^{2/3}/20}{n^{1/3}}} {\binom{n}{n^{1/3}}}   \geq \left(\frac{n-n^{1/3}-n^{2/3}/20}{n-n^{1/3}}\right)^{n^{1/3}} \geq (1- n^{-1/3}/10)^{n^{1/3}} \geq   1 - 1/10 .
\]
We also have
\[
\Pr[W_{1,2}'] = \Pr[W_{1,3}'] \geq 1- \frac{\binom{n-5n^{2/3}}{n^{1/3}}} {\binom{n}{n^{1/3}}} \geq 1-\left(\frac{n-5n^{2/3}}{n}\right)^{n^{1/3}} \geq 1- e^{-5}.
\]
For $W_{1,4}'$, define $X_i$ to be the indicator variable of whether $i \not\in S$ (i.e. $X_i = 1$ if $i \not\in S$ and $X_i = 0$ otherwise). We have that for any subset $N' \subseteq N$, $\Pr[\bigwedge_{i\in N'} X_i = 1] \leq (1-1/n^{2/3})^{|N'|}$. By generalized Chernoff bound (Theorem \ref{thm:gcb}), we have
\begin{align*}
&\Pr[|S \cap \{k-40\cdot l_i \cdot n^{2/3}+1,..., k\}| < l_i] \\
&= \Pr\left[ \sum_{a \in \{k-40\cdot l_i \cdot n^{2/3}+1,..., k\}} X_a \geq \left(1- \frac{1}{40n^{2/3}}\right) \cdot 40\cdot l_i \cdot n^{2/3}\right] \\
&\leq \exp\left(-40\cdot l_i \cdot n^{2/3} \cdot \D_e\left(1- \frac{1}{40n^{2/3}} \| 1- \frac{1}{n^{2/3}}\right)\right).\\
\end{align*}
By Fact \ref{fact:kl}, we have
\[
\D_e\left(1- \frac{1}{40n^{2/3}} \| 1- \frac{1}{n^{2/3}}\right) \geq \left(\frac{1}{n^{2/3}} - \frac{1}{40n^{2/3}}\right)^2/ \frac{2}{n^{2/3}} \geq \frac{19}{40n^{2/3}}.
\]
Therefore
\[
\Pr[|S \cap \{k-40\cdot l_i \cdot n^{2/3}+1,..., k\}| < l_i] \leq \exp(-19l_i).
\]
By union bound
\[
\Pr[W_{1,4}' ]\geq 1- 2\sum_{i=1}^{\log^*(n)} \exp(-19l_i) \geq 1-e^{-16}. 
\]
By union bound again,
\[
\Pr[W_1'] \geq 1- (1-\Pr[W_{1,1}']) - (1-\Pr[W_{1,2}'])- (1-\Pr[W_{1,3}']) - (1-\Pr[W_{1,4}']) \geq  1 - 1/10 - 2/e^5- 1/e^{16}.
\]

For (iii), we are going to fix the randomness of (i) and (ii) and condition on $W_{1,1} \cap W_1'$. By $W'_{1,2}$, $S \cap \{ k - 6 n^{2/3},...,k-n^{2/3}/40-1\}$ is not empty and set $J$ such that $s_J$ has the lowest rank in $S \cap \{ k - 6 n^{2/3},...,k-n^{2/3}/40-1\}$. By $W'_{1,1}$ and $W'_{1,3}$, we know that $s_{J+1} \in  \{ k + n^{2/3}/40+1,...,k+6 n^{2/3}\}$. Consider event $W_1''$ as the intersection of following events:
\begin{itemize}
\item $W''_{1,1}$: $|s_J - C_J| \leq n^{2/3}/40$ and $|s_{J+1} - C_{J+1}| \leq n^{2/3}/40$.
\item $W''_{1,2}$: $\forall i = 1,...,\log^*(n)$, $|s_{J-l_i} -C_{J-l_i}| \leq l_i \cdot n^{2/3}$ and $|s_{J+1 + l_i} -C_{J+1+ l_i}| \leq l_i \cdot n^{2/3}$.
\item $W''_{1,3}$: $\forall i = 1,...,\log^*(n)-1$, $a \in \{k-6 l_i^5+1,..., k+ 6l_i^5\} $, $|p(a)-J| \leq l_{i+1}$. 
\item $W''_{1,4}$: $\forall i \in N$, suppose $s_j \leq i < s_{j+1}$, then $|p(i) - j| \leq l_1$. 
\end{itemize}

Now we analyze the probability $W_1''$ holds (i.e. $\Pr[W_1'']$). For $W''_{1,1}$, by Lemma \ref{lem:cj}, we have
\[
\Pr[W''_{1,1}|W_{1,1} \cap W_1'] \geq 1 - 2 \cdot \frac{1}{5000} \cdot 40 \geq 1 - 1/50.
\]
For $W''_{1,2}$, by Lemma \ref{lem:cj} again, we have
\[
\Pr[W''_{1,2}|W_{1,1} \cap W_1'] \geq 1- 2\cdot \sum_{i=1}^{\log^*(n)}   \frac{1}{5000l_i} \geq 1-1/500.
 \]
 For $W''_{1,3}$, by Lemma \ref{lem:brw}, we have
 \[
\Pr[W''_{1,3}|W_{1,1} \cap W_1'] \geq 1- \sum_{i=1}^{\log^*(n)-1} 2 \cdot6 \cdot l_i^5 \cdot 3 \cdot e^{-15l_{i+1}} \geq 1- \sum_{i=1}^{\log^*(n)-1} \frac{1}{e^5 \cdot l_i} \geq 1-e^{-4}.
 \]
 For $W''_{1,4}$, by Lemma \ref{lem:brw} again, we have
 \[
\Pr[W''_{1,4}|W_{1,1} \cap W_1'] \geq 1- n \cdot 3 \cdot e^{-15l_1} \geq 1- 3/n^{14}.
 \]
 By union bound, we have
 \begin{align*}
 &~~~\Pr[W_1''|W_{1,1} \cap W_1'] \\
 &\geq 1- (1-\Pr[W_{1,1}''|W_{1,1} \cap W_1']) - (1-\Pr[W_{1,2}''|W_{1,1} \cap W_1'])\\
 &~~~- (1-\Pr[W_{1,3}''|W_{1,1} \cap W_1']) - (1-\Pr[W_{1,4}''|W_{1,1} \cap W_1']) \\
 &\geq  1 - 1/50 -1/500- 1/e^4- 3/n^{14}.
 \end{align*}
 Then we have
 \begin{align*}
 \Pr[ W_{1,1} \cap W_1' \cap W_1''] &\geq  \Pr[ W_{1,1} \cap W_1'] \cdot \Pr[W_1''|W_{1,1} \cap W_1'] \\
 &\geq 1- 1/n - 1/10 - 2/e^5- 1/e^{16} - 1/50 -1/500- 1/e^4- 3/n^{14} \\
 &\geq 1-1/6.
 \end{align*}
Finally we show that conditioned on $W_{1,1} \cap W_1'$, $W_1''$ implies $W_1$. 
\begin{itemize}
\item $W_{1,2}$ and $W_{1,3}$: By $W_{1,1}''$, we know that $\sum_{j<J} |P_j| = C_J < k$ and $\sum_{j<J+1}|P_j| = C_{J+1}  \geq k$. Therefore $m = J$ and then $W_{1,2}$ and $W_{1,3}$ hold. 
\item $W_{1,4}$: By $W''_{1,2}$ and $W'_{1,4}$, we know that $C_{m-l_i} \geq k-41\cdot l_i \cdot n^{2/3}$ and $C_{m+l_i+1} \leq k+41\cdot l_i \cdot n^{2/3}v$. Then $|\bigcup_{j=m-l_i}^{m+l_i} P_j| = C_{m+l_i+1} - C_{m-l_i} \leq 100\cdot l_i \cdot n^{2/3}$. 
\item $W_{1,5}$: By $W_{1,3}''$, we know that, $\forall i = 1,...,\log^*(n)-1$, $a \in \{k-6 l_i^5+1,..., k+ 6 l_i^5\} $, $|p(a)-m| \leq l_{i+1}$ and therefore $a\in \bigcup_{j=m-l_{i+1}}^{m+l_{i+1}} P_j$.
\item $W_{1,6}$: By $W_{1,4}''$, we know that $\forall i \leq k$, $p(i) \leq k+ l_1$ and therefore $i \not\in \bigcup_{j> m+l_1} P_j $. This means $\bigcup_{j> m+l_1} P_j  \subseteq \{k+1,...,n\}$. The same argument will also give $\bigcup_{j< m-l_1} P_j \subseteq \{1,...,k\}$. 
\end{itemize}

Therefore $\Pr[W_1] \geq \Pr[W_{1,1} \cap W_1'\cap W_1''] \geq 1- 1/6$. 
\end{proof}

\subsubsection{Second Round of the 2-Round Algorithm}
Now we are going to assume $W_1$ holds and analyze the second round of the algorithm. Consider event $W_2$ to be: 
\begin{itemize}
\item $W_{2,1}$: $\forall i = 1,...,\log^*(n)$, $N_i \subseteq \bigcup_{j=m-l_i}^{m+l_i} P_j $. 
\item $W_{2,2}$: $\forall i = 2,...,\log^*(n)$, $N_i \subseteq \{k-6l_i^5+1,..., k+6l_i^5\} $. 
\item $W_{2,3}$: $\forall i = 1,...,\log^*(n)$, $A_i \subseteq \{1,...,k\}$ and $B_i \subseteq \{k+1,...,n\}$. 
\item $W_{2,4}$: $N_{\log^*(n)+1} = \emptyset$. 
\end{itemize}

We use $W_2$ to indicate the success of the second round. We show in Lemma \ref{lem:2r-2r} that conditioned on $W_1$, $W_2$ happens with probability at least $1-1/6$.

\begin{lemma}
\label{lem:2r-2r}
$\Pr[W_2|W_1] \geq 1- 1/6$. 
\end{lemma}

\begin{proof}
Assume $W_1$ holds. We first show that the algorithm does not halt. The first round of the algorithm uses $c_1  n^{4/3} + 10n^{2/3}\log(n)$ comparisons. In the second round, by $W_{1,4}$, the number of comparisons used by the algorithm is 
\[
\sum_{i=1}^{\log^*(n)}\left|\bigcup_{j=m-l_i}^{m+l_i} P_j\right| \cdot\left( \left|\bigcup_{j=m-l_i}^{m+l_i} P_j\right|/ l_i^4\right) \cdot c_2 l_i \leq 10000 \cdot c_2  \cdot \sum_{i=1}^{\log^*(n)} \frac{1}{l_i} \cdot n^{4/3} \leq 20000\cdot c_2 \cdot n^{4/3}.
\]
Therefore picking $c_0$ to be a large enough constant will make sure that the algorithm does not halt conditioned on $W_1$.

Consider event $W_2'$ to be the intersection of the following two events:
\begin{itemize}
\item $W_{2,1}'$: $\forall i = 1,...,\log^*(n)-1$, $|S_i \cap \{k- 6l_i^5+1,..., k\}| >0 $ and  $|S_i \cap \{k+ 1,..., k+6l_i^5\}| >0 $.
\item $W_{2,2}'$: $\forall i = 1,...,\log^*(n)$ and each pair of item $j, j'$ ranks in $\{ k-6l_{i-1}^5+1, ..., k+6l_{i-1}^5 \} \cap \bigcup_{j=m-l_{i}}^{m+l_{i}} P_j$ and $j \in S_i$, if $j$ and $j'$ are compared in the $i$-th iteration, then whether $j$ beats $j'$ in the $i$-th iteration of the second round is consistent with the true ordering.
\end{itemize}

Now let's analyze $\Pr[W_2'|W_1]$. We start with $W_{2,1}'$. By $W_{1,5}$, we know that $\forall i = 1,...,\log^*(n)-1$, we have $\{k- 6l_i^5+1,..., k+6l_i^5\}\subseteq \bigcup_{j=m-l_{i+1}}^{m+l_{i+1}} P_j$. Therefore,
\[
\Pr[W_{2,1}'|W_1] \geq 1- 2\sum_{i=1}^{\log^*(n)-1} (1-1/l_i^4)^{6l_i^5} \geq 1 -2\sum_{i=1}^{\log^*(n)-1}e^{-6l_i} \leq 1 -4 e^{-6l_{\log^*(n)-1}} \leq 1-4e^{-6}.
\]

In the $i$-th iteration, for each pair of item $j, j'$ ranks in $\{ k-6l_{i-1}^5+1, ..., k+6l_{i-1}^5 \}\cap \bigcup_{j=m-l_{i}}^{m+l_{i}} P_j$ and $j \in S_i$, if $j$ and $j'$ are compared in the $i$-th iteration, then there are $c_2 \cdot l_i$ comparisons between $j$ and $j'$. By Chernoff bound, the majority of these comparisons differs from the true ordering with probability at most 
\[
\exp\left(-\frac{2}{3} c_2 \cdot l_i \cdot \left(\frac{1}{4}\right)^2 \cdot \frac{1}{2}\right) = \exp(-c_2 \cdot l_i/ 48). 
\]
By union bound, 
\begin{align*}
1-\Pr[W_{2,2}'|W_1] &\leq \sum_{i=1}^{\log^*(n)} (12l_{i-1}^6)^2 \cdot \exp(-c_2 \cdot l_i/ 48)\\
& \leq \sum_{i=1}^{\log^*(n)} (12l_{i-1}^6)^2 \cdot \frac{1}{2^{11} \cdot l_{i-1}^7 } \leq 2 \cdot \frac{1}{2^5 \cdot  l_{log^*{n}-1} } \leq \frac{1}{16}.
\end{align*}
By union bound again,
\[
\Pr[W_2'|W_1] \geq  1- (1-\Pr[W_{2,1}']) - (1-\Pr[W_{2,1}'])  \geq 1- 1/6.
\]

Now we are going to show that $W_1$ and $W_2'$ imply $W_2$. Define $W^i_{2,1}$ to be the event that $N_i \subseteq \bigcup_{j=m-l_i}^{m+l_i} P_j $, $W^i_{2,2}$ to be the event that $N_i \subseteq \{k-6l_i^5+1, ...,k+6l_i^5\} $ and $W^i_{2,3}$ to be the event that $A_i \subseteq \{1,...,k\}$ and $B_i \subseteq \{k+1,...,n\}$. Now assume $W_1$ and $W_2'$ hold. 
\begin{itemize}
\item $W^i_{2,1} \Rightarrow W^{i+1}_{2,2}$, $W^i_{2,1} \Rightarrow W^i_{2,3}$: By $W^i_{2,1}$, we know that all the items in $N_i$ are compared to all the items in $N_i \cap S_i$ in the $i$-th iteration. Then by $W'_{2,2}$ we know that, $\forall j \in S_i \cap N_i$, $r_i(j)$ is $j$'s correct rank in $N_i$. By $W_{2,1}$, we know that either $a_i \in \{k-6l_i^5+1,..., k\}$ or we have $a_i$ does not exist and $N_i \cap \{1,...,k-6l_i^5\} = \emptyset$. In both cases, we have $A_i \subseteq \{1,...,k\}$ and $N_{i+1} \cap \{1,...,k-6l_i^5\} = \emptyset$. Similarly, we get $B_i \subseteq \{k+1,...,n\}$ and $N_{i+1} \cap \{k+6l_i^5+1, ...,n\} = \emptyset$. So we get $W^{i+1}_{2,2}$and $W^i_{2,3}$.
\item $W^i_{2,2} \Rightarrow W^{i}_{2,1}$: By $W_{1,5}$, we know that $ \{k-6l_i^5+1, k+6l_i^5\} \subseteq \bigcup_{j=m-l_i}^{m+l_i} P_j$. Then $W^i_{2,2}$ simply implies $W^i_{2,1}$. 
\end{itemize}
Since $W^1_{2,1}$ holds, by induction we know that $W_{2,1}$, $W_{2,2}$ and $W_{2,3}$ hold. Also notice that $l_{\log^*(n)} = 1$ and thus $N_{\log^*(n)}\subseteq S_{\log^*(n)}$. Since $W^{\log^*(n)}_{2,1}$ implies $r_{\log^*(n)}(j)$ is $j$'s correct rank in $N_{\log^*(n)}$ for $\forall j \in N_{\log^*(n)}$, we know that $a_{\log^*(n)}$ and $b_{\log^*(n)}$ will rank at $k_{\log^*(n)}$ and $k_{\log^*(n)}+1$ in $N_{\log^*(n)}$. And therefore we have $N_{\log^*(n)+1} = \emptyset$ (event $W_{2,4}$). 

To sum up, we have
\[
\Pr[W_2|W_1] \geq \Pr[W_2'|W_1] \geq 1-1/6.
\]

\end{proof}

\section{Sorted Top-$k$ in the Noisy Case}
\label{sec:stopk_n}
In this section, we consider sorted top-$k$ in the noisy case. In particular, we show that the sample complexity of 1-round algorithms is $\Theta(n^2 \log(k))$ (in Section \ref{sec:stopk1_n}) and the sample complexity of 2-round algorithms is $\Theta(n^{4/3} \log(k))$ (in Section \ref{sec:stopk2_n}). 

When $k=1$, sorted top-$k$ is equivalent to top-$k$. In this section, we only consider cases when $k>1$.

\subsection{1-Round Sorted Top-$k$ in the Noisy Case}
\label{sec:stopk1_n}

In this sub-section, we show an 1-round algorithm in Lemma \ref{lem:1stopk_noisy_ub} and a matching lower bound in Lemma \ref{lem:1stopk_noisy_lb}.

\begin{algorithm}[H]
    \caption{$1$-round algorithm for sorted top-$k$ with noisy comparisons}
    	\label{alg:1stopk_noisy}
    \begin{algorithmic}[1]
    	\STATE Run 3 copies of Algorithm \ref{alg:1topk} (1-round top-$k$ algorithm) in parallel.
	\STATE In the same round, compare each pair $100(\log(k)+1)$ time.
 	\STATE Let $S$ be the majority answer of the 3 copies of Algorithm \ref{alg:1topk}. 
	\STATE For each item $i,j \in S$, say $i$ beats $j$ if $i$ wins the majority of comparisons between $i$ and $j$ in step 2. Rank items in $S$ in the decreasing order of how many items they beat (break tie arbitrarily). Output this sorted list of $S$. 
     \end{algorithmic}
\end{algorithm}

\begin{lemma}
\label{lem:1stopk_noisy_ub}
We have an 1-round algorithm (Algorithm \ref{alg:1stopk_noisy}) which solves sorted top-$k$ in the noisy case with $O(n^2\log(k) )$ comparisons.
\end{lemma}

\begin{proof}
First of all, by the definition of Algorithm \ref{alg:1stopk_noisy} and Lemma \ref{lem:1topk} (which bounds the number of comparisons of Algorithm \ref{alg:1topk}), we know Algorithm \ref{alg:1stopk_noisy} uses  $O(n^2\log(k) )$ comparisons.

By Lemma \ref{lem:1topk}, we know that the probability that Algorithm \ref{alg:1topk} outputs top-$k$ correctly with probability at least $2/3$. Therefore $S$ is the set of top-$k$ items with probability at least $1-7/27$. 

For each pair of items $(i,j)$, the probability that whether $i$ beats $j$ is consistent with their underlying order is at least $1-\frac{1}{27k^2}$. For a fixed $S$, by union bound, 
\[
\Pr[\forall i,j \in S, i \text{ beats } j \text{ is consistent with their underlying order}] \geq 1-1/27.
\]

By union bound again, Algorithm \ref{alg:1stopk_noisy} is correct with probability at least $1-7/27-1/27 > 2/3$.

\end{proof}

\begin{lemma}
\label{lem:1stopk_noisy_lb}
Any 1-round algorithm needs $\Omega(n^2\log(k) )$ comparisons to output sorted top-$k$ correctly with probability at least $2/3$ in the noisy case. 
\end{lemma}

\begin{proof}
For $k \leq 36$, the lemma is directly implied by 1-round lower bound of top-$k$ in the noisy case. In the rest of the proof, we assume $k >36$.

Consider any algorithm $A$ with fewer than $w =\frac{1}{100}\cdot n^2 \log(k)$ comparisons. We are going to show $A$ outputs sorted top-$k$ incorrectly with probability $>1/3$. 

Let's assume $A$ labels items as $1,...,n$. Let $\Pi(i)$ be the random variable of the actual rank of item with label $i$, $\forall i \in [n]$. Before the algorithm makes any comparisons, $\Pi$ distributed as a uniform distribution over all permutations of $[n]$. Wlog we can assume $A$ is deterministic. We use $H$ to denote the random variable of the comparison results and $h$ to denote the realized value of the comparison results. We use $A(h)$ to denote the sorted top-$k$ outputted by algorithm $A$ given comparison results $h$.

For any ranking $\pi$, define $\pi^{i,j}$ as the following:
\begin{itemize}
\item $\pi^{i,j} (l) = \pi(l)$ if $l \neq i,j$.
\item $\pi^{i,j}(i) = \pi(j)$.
\item $\pi^{i,j} (j) = \pi(i)$. 
\end{itemize}

For $l = 1,...,k-1$ and $\pi$ to be any permutation of $[n]$, define $s(l, \pi)$ to be the number of comparisons between items with label $\pi^{-1}(l)$ and $\pi^{-1}(l+1)$. Define $S(\pi) = \{l | s(l,\pi) \leq \log(k)/2, l \in [k-1]\}$. We know that for any $l \in [k-1]$, $\E[s(l,\Pi)] \leq w / \binom{n}{2} \leq \log(k)/ 8$.  Therefore, by Markov's inequality, we know that with probability at least $1/2$, $\sum_{l=1}^{k-1} s(l,\Pi) \leq (k-1) \log(k)/4 $. When $\sum_{l=1}^{k-1} s(l,\Pi) \leq (k-1) \log(k)/ 4$, we know that $S(\pi) \geq (k-1)/2 \geq k/3$. Therefore $\Pr[S(\Pi) \geq k/3] \geq 1/2$.

When $|\pi(i) - \pi(j) | = 1$ and there are $s$ comparisons between items with labels $i$ and $j$ in algorithm $A$, we have
\[
\frac{\Pr[H = h| \Pi = \pi]}{\Pr[H = h|\Pi = \pi^{i,j}]} \leq \left(\frac{2/3}{1/3} \right)^s = 2^s.
\]
To sum up, we have
\begin{align*}
&\Pr[A \text{ outputs correctly}]\\
 = &\sum_{\pi} \sum_{h: A(h) =\text{sorted top-}k} \Pr[\Pi = \pi] \cdot \Pr[ H = h| \Pi = \pi].  \\
= &\Pr[|S(\Pi)| < k/3] + \sum_{\pi:|S(\pi)| \geq k/3}~~ \sum_{h: A(h) =\text{sorted top-}k} \Pr[\Pi = \pi] \cdot \Pr[ H = h| \Pi = \pi].  \\
\leq &1/2 +\sum_{\pi: |S(\pi)| \geq k/3}~~ \sum_{h: A(h) =\text{sorted top-}k(\pi)} \frac{1}{|S(\pi)|} \sum_{l \in S(\pi), \pi' = \pi^{\pi^{-1}(l), \pi^{-1}(l+1)}} \\
&~~~~~~\Pr[\Pi = \pi'] \cdot \Pr[ H = h| \Pi = \pi']  \cdot\frac{\Pr[H = h| \Pi = \pi]}{\Pr[H = h|\Pi = \pi']}   \\
\leq &1/2+\sum_{\pi'}  \sum_{h: A(h) =\text{sorted top-}k(\pi')} \Pr[\Pi = \pi'] \cdot \Pr[ H = h| \Pi = \pi']\cdot 2^{\log(k)/2} \cdot \frac{3}{k} \\
<& 1/2+ \Pr[A \text{ outputs incorrectly}]/2.
\end{align*}
Therefore $\Pr[A \text{ outputs incorrectly}] > 1/3$. 
\end{proof}

\subsection{2-Round Sorted Top-$k$ in the Noisy Case}
\label{sec:stopk2_n}

In this sub-section, we show a 2-round algorithm in Lemma \ref{lem:2stopk_noisy_ub} and a matching lower bound in Lemma \ref{lem:2stopk_noisy_lb}.

\begin{algorithm}[H]
    \caption{$2$-round algorithm for sorted top-$k$ with noisy comparisons (for $k<n^{1/10}$)}
    	\label{alg:2stopk_noisy}
    \begin{algorithmic}[1]
	\STATE Randomly partition $N$ into $n^{2/3}$ sets of size $n^{1/3}$: $S_1,...,S_{n^{2/3}}$. 
	\STATE Round 1: For each $S_i$, we run $200\log(k)$ copies of the 1-round top-1 algorithm (Algorithm \ref{alg:1topk}) to find the top-1 of $S_i$ in parallel and take the majority answer to be item $t_i$.
	\STATE Round 2: Let $T$ be the set of all $t_i$'s. We have $|T| = n^{2/3}$. Run the 1-round sorted top-$k$ algorithm (Algorithm \ref{alg:1stopk_noisy}) to find the sorted top-$k$ of $T$ and output it.  
     \end{algorithmic}
\end{algorithm}

\begin{lemma}
\label{lem:2stopk_noisy_ub}
We have a 2-round algorithm which solves sorted top-$k$ in the noisy case with $O((n^{4/3} + n \sqrt{k})\log(k))$ comparisons.
\end{lemma}

\begin{proof}
When $k \geq n^{1/10}$, we have $\log(k) = \Theta(\log(n))$. We just simply use the 2-round sorted top-$k$ algorithm (Algorithm \ref{alg:rsorted1}) in the noiseless case and turn it into a 2-round sorted top-$k$ algorithm in the noisy case by repeating each comparison $\Theta(\log(n))$ times. The algorithm uses  $O((n^{4/3} + n \sqrt{k})\log(n)) = O((n^{4/3} + n \sqrt{k})\log(k))$ comparisons.

When $k < n^{1/10}$, we use Algorithm \ref{alg:2stopk_noisy}. It's clear that this algorithm uses $O(n^{4/3} \log(k)) = O((n^{4/3} + n \sqrt{k})\log(k))$ comparisons. We are going to show that this algorithm succeeds with probability $2/3$. Consider the following events:

\begin{itemize}
\item No two items in top-$k$ are placed in the same $S_i$. This event happens with probability at least $1-\frac{k^2}{n^{2/3}} \geq 1 - n^{1/5-2/3}$. 
\item For each $S_i$ that contains a top-$k$ item, $t_i$ is actually the top-1 of $S_i$. This event happens with probability $1-  k \cdot \frac{1}{40k} > 1-1/40$. 
\item In the second round, the output is the correct sorted top-$k$ of $T$. This happens with probability at least $1-8/27$. 
\end{itemize}

When all these events happen, it's easy to check that Algorithm \ref{alg:2stopk_noisy} outputs correctly. By union bound, all of these events happen with probability at least $1-n^{1/5-2/3}-1/40-8/27 > 2/3$. 
\end{proof}

\begin{lemma}
\label{lem:2stopk_noisy_lb}
Any 2-round algorithm needs $\Omega((n^{4/3} + n \sqrt{k})log(k))$ comparisons to output sorted top-$k$ correctly with probability at least $2/3$ in the noisy case. 
\end{lemma}

\begin{proof}
For notation convenience, we will wlog assume $k$ is even. For odd $k$'s, because sorted top-$(k-1)$ is an easier task than sorted top-$k$, we will just apply the lower bound of sorted top-$(k-1)$. 

Consider some algorithm $A$ that uses fewer than $c \cdot \max(n^{4/3}, n \sqrt{k})log(k)$ comparisons for $c = \frac{1}{10^4}$. We are going to show that $A$ outputs sorted top-$k$ incorrectly with probability $>1/3$. Let's assume $A$ labels items as $1,...,n$. Let $\Pi(i)$ be the actual ranking of item with label $i$, $\forall i \in [n]$. Before the algorithm makes any comparisons, $\Pi$ distributed as a uniform distribution over all permutations of $[n]$. Wlog we can assume $A$ is deterministic. We use $H_1$ to denote the random variable of the comparison results in the first round and $H_2$ to denote the random variable of the comparison results in the second round. We use $H = (H_1, H_2)$ to denote the full history.

Define $W$ to be the event that there exists $k^{1/4}$ pairs of items $(u,v)$ such that $|\Pi(u) - \Pi(v)| = 1$, $\Pi(u) \leq k$ , $\Pi(v)\leq k$ and $(u,v)$ are compared fewer than $0.2 \log(k)$ times in $H$. We will first prove $\Pr[W] \geq 2/3$, and then we will show that this implies $A$ outputs incorrectly with probability at least $>1/3$. 

We first prove $\Pr[W] \geq 2/3$. We consider two different cases depending on how large $k$ is. 

\begin{itemize}
\item Case 1: 
$k \geq n^{2/3}$. Set $m = \frac{k^{3/2}}{n}$. Define $S_i$ to be the set of items ranked in $\{(i-1) \cdot \frac{n}{\sqrt{k}} + 1,..., i \cdot \frac{n}{\sqrt{k}} \}$ for $i = 1,...,m$. $S_i$'s are random variables depending on $\Pi$. Let $Q$ be the collection of $S_1,...,S_m$ and $\Pi^{-1}(1),\Pi^{-1}(3),...,\Pi^{-1}(k-1)$ (these are labels of top-$k$ items with odd rankings). 

For $i \in [m]$, define $G_i$ to be the set of items in $S_i$ which satisfy the following:
\begin{itemize}
\item Have even rankings.
\item Are compared to at most $0.1\log(k)$ items in $S_i$ in the first round of $A$.
\end{itemize}
Notice that once $Q$ is fixed, all $G_i$'s are fixed. 

For $i \in [m]$, define $U_i$ to be the set of pairs of items $u,v$ such that
\begin{itemize}
\item $(u,v)$ are compared more than $0.1 \log(k)$ times in the second round of $A$. 
\item $u, v \in S_i$.
\item $u$ has an even ranking and $v$ has an odd ranking.
\end{itemize}
Notice that once $Q$ and $H_1$ are fixed, $U_i$'s are fixed.

Finally for $i \in [m]$, define $X_i$ be the number of pairs $(u,v)$ such that
\begin{itemize}
\item $u,v \in S_i$.
\item $u,v$ are compared at most $0.2 \log(k)$ times in the two rounds of $A$. 
\item There exists $l$ such that $\Pi(u) = 2l-1$ and $\Pi(v) = 2l$. 
\end{itemize}

Clearly, $(\exists i, X_i \geq k^{1/4})$ implies $W$, in other words, $\Pr[W] \geq \Pr[\sum_{i=1}^m X_i \geq k^{1/4}]\geq \Pr[\exists i, X_i \geq k^{1/4}]$. So it suffices to prove that $\Pr[\exists i, X_i \geq k^{1/4}] \geq 2/3$. Also notice that, once $Q$ is fixed, $X_i$'s are independent. 

Define $W_G$ to be the event in which there are at least $2m/3$ $i$'s with $|G_i| \geq \frac{3n}{8\sqrt{k}}$. Notice that once $Q$ is fixed, $W_G$ is also fixed. We want to show that $\Pr[W_G] \geq 5/6$. First of all, each comparison in the first round is a comparison between two items in the same $S_i$ with probability at most $\frac{k}{n} \cdot \frac{1}{\sqrt{k}}$. So with probability at least $5/6$, the number of such comparisons is at most $6c \cdot k \log k$, as there are at most $c \cdot  n \sqrt{k}log(k)$ comparisons in the first round. In this case, consider the top-$k$ items with even rankings. Since $6c \cdot 2 \leq \frac{1}{240}$, we know that at most $ k/24$ of them are compared to items in the same $S_i$ more than $0.1\log(k)$ times. Therefore, at most $m/3$ $G_i$'s can have $|G_i| \leq \frac{n}{2\sqrt{k}} - \frac{n}{8\sqrt{k}} = \frac{3n}{8\sqrt{k}}$. To sum up, we have $\Pr[W_G] \geq 5/6$.

Now consider $U_i$'s. We know that there are at most $c \cdot n \sqrt{k} \log(k)$ comparisons in the second round, therefore we have $\sum_{i=1}^m |U_i| \leq \frac{c \cdot n \sqrt{k} \log(k)}{0.1\log(k)} =10 c \cdot n \sqrt{k}$. Since $c \leq \frac{1}{10 \cdot 3 \cdot 32}$, at most $m/3$ $U_i$'s can have $|U_i| \geq \frac{n^2}{32k}$.  

Now we fix $Q$ and $H_1$ such that $W_G$ happens. We know that there are at least $m/3$ $i$'s satisfying both $|G_i| \geq \frac{n}{4\sqrt{k}}$ and $|U_i| \leq \frac{n^2}{32k}$. Consider any such $i$. Define $B_i$ to be the set of items such that for each $u \in B_i$, 
\begin{itemize}
\item $u \in G_i$.
\item There are at most $\frac{n}{8\sqrt{k}}$ $v$'s such that $(u,v) \in U_i$. 
\end{itemize}
By averaging argument, we have $|B_i| \geq |G_i| - |U_i| \cdot \frac{8\sqrt{k}}{n} \geq \frac{n}{8\sqrt{k}}$. Now we throw away arbitrary items in $B_i$ to make it has size exactly $\frac{n}{8\sqrt{k}}$. 

For each item $u \in B_i$, define $Y_u$ to be 1 if $(u, \Pi^{-1}( \Pi(u) - 1)) \not \in U_i$, otherwise $Y_u$ is 0. Here $\Pi^{-1}( \Pi(u) - 1))$ is just the item ranks right above $u$ in $\Pi$. We have $X_i \geq \sum_{u \in B_i} Y_u$. Since we have already fix $Q$, so the items with odd rankings in $S_i$ have fixed rankings. Therefore $Y_u$ only depends on $\Pi(u)$. 

Now consider any subset $S \subseteq B_i$ and any $u \in S$. Define $S' = S \backslash \{u\}$. We fix $\Pi(v)$ for all $v$ in $S'$. Now consider the probability such that $Y_u = 1$. For any $\pi$ such that $Y_u = 0$, consider $u'$ such that $u' \in B_i$, $u' \not \in S'$ and $(u, \pi^{-1}(\pi(u)-1)) \not \in U_i$. We have at least $ \frac{3n}{8\sqrt{k}}- \frac{n}{8\sqrt{k}} - \frac{n}{8\sqrt{k}} =\frac{n}{8\sqrt{k}}$ many such $u'$. Define a different ranking $\pi^{(u,u')}$ such that it is the same as $\pi$ except the rankings of $u$ and $u'$ are swapped. Notice that with the same $H_1$, switching from $\Pi= \pi$ to $\Pi = \pi^{u,u'}$ does not change $Q$ and $\Pi(v)$ for all $v$ in $S'$, but changes $Y_u$ from 0 to 1.  We know that each of $u$ and $u'$ is compared to items in $S_i$ at most $0.1 \log(k)$ times in the first round of $A$, we have
\[
\frac{\Pr[H_1 =h_1 | \Pi = \pi]}{\Pr[H_1=h_1|\Pi = \pi^{u,u'}]} \leq \left(\frac{\frac{2}{3}}{\frac{1}{3}} \right) ^{0.2\log(k)} \leq k^{0.2}.
\]
As $\Pr[\Pi = \pi ] = \Pr[\Pi = \pi^{u,u'}]$, we have 
\[
\frac{\Pr[\Pi = \pi|H_1 = h_1]}{\Pr[\Pi = \pi^{u,u'}|H_1 = h_1]} = \frac{\Pr[H_1 =h_1 | \Pi = \pi]}{\Pr[H_1=h_1|\Pi = \pi^{u,u'}]} \leq k^{0.2}.
\]
We know that each $\pi$ corresponds to at least $ \frac{n}{8\sqrt{k}}$ such $\pi^{u,u'}$ and each $\pi^{u,u'}$ corresponds to at most $\frac{n}{2\sqrt{k}}$ such $\pi$. Therefore, 
\begin{align*}
\Pr[Y_u = 0| Q,H_1, \Pi(v) ~~\forall v \in S']  \leq \frac{ \frac{n}{8\sqrt{k}}}{\frac{n}{2\sqrt{k}} \cdot k^{0.2}}  \cdot \Pr[Y_u = 1| Q,H_1, \Pi(v) ~~\forall v \in S']
\end{align*}
So we have
\[
\Pr[Y_u = 0| Q,H_1, \Pi(v) ~~\forall v \in S']  \leq 1- \frac{1}{4k^{0.2}}.
\]
And this implies
\[
\Pr[Y_u = 0| Q,H_1, Y_v = 0~~ \forall v \in S']  \leq 1- \frac{1}{4k^{0.2}}.
\]
For any $S \subseteq B_i$, 
\[
\Pr\left[\bigwedge_{u \in S} Y_u = 0|Q,H_1\right] \leq \left(1-\frac{1}{4k^{0.2}}\right)^{|S|}.
\]
Then we have 
\begin{align*}
&\Pr\left[\sum_{u \in B_i} Y_u \geq k^{1/4} \right] \\
\geq &\Pr\left[ \sum_{u \in B_i} Y_u \geq \frac{n}{8\sqrt{k}} \cdot \frac{1}{4k^{0.2}} \cdot \frac{1}{2}\right] ~~~\text{(for large enough } n\text{)}\\
\geq &\exp\left(-\frac{n}{4\sqrt{k}} \cdot \frac{1}{4k^{0.2}} \cdot \frac{1}{2} \cdot \left(\frac{1}{2}\right)^2\right) ~~~\text{(by Theorem \ref{thm:gcb} and Fact \ref{fact:kl})}\\
\geq & 5/6 ~~~\text{(for large enough } n\text{)}\\
\end{align*}

To sum up, we have
\begin{align*}
\Pr[W] &\geq  \Pr\left[\exists i,  X_i \geq k^{1/4}\right] \\
&\geq \sum_{q,h_1: W_G\text{ happens}} \Pr[Q= q,H_1 =h_1] \cdot \Pr\left[\exists i, X_i \geq k^{1/4}|Q=q,H_1 =h_1\right]\\
&\geq \Pr[W_G] \cdot \frac{5}{6}\\
&\geq \frac{5}{6} \cdot \frac{5}{6} > \frac{2}{3}.
\end{align*}
\item  Case 2: $k < n^{2/3}$. We are going to use a similar proof strategy as Case 1 (with slight changes). 

Define $Q$ to be the set of top-$n^{2/3}$ items and define $G$ to be the set of items in $Q$ such that for each $u \in G$, $u$ is compared to at most $0.1\log(k)$ items in $Q$ in the first round of $A$. Notice that once $Q$ is fixed, $G$ is also fixed. 

Define $U$ to be the set of pair $(u,v)$ such that
\begin{itemize}
\item $(u,v)$ are compared more than $0.1 \log(k)$ times in the second round of $A$.
\item $u,v \in Q$. 
\end{itemize}
$U$ is fixed once $Q$ and $H_1$ are fixed.
Finally for $l = 1,..., k/2$, define $Y_l$ to be 1 if the following is true (otherwise $Y_l = 0$):
\begin{itemize}
\item Let $u = \Pi^{-1} (2l-1)$ and $v = \Pi^{-1}(2l)$. 
\item $u,v \in G$.
\item $(u,v) \not \in U$. 
\end{itemize}
Clearly, $\sum_{l=1}^{k/4} Y_l \geq k^{1/4}$ implies $W$. In other words, $\Pr[W]  \geq \Pr[ \sum_{l=1}^{k/4} Y_l \geq k^{1/4}]$. It suffices to prove $ \Pr[ \sum_{l=1}^{k/2} Y_l \geq k^{1/4}] \geq 2/3$. 

Define $W_G$ to be the event that $|G| \geq \frac{3}{4} n^{2/3}$. First of all, each comparison in the first round is a comparison between two items in $Q$ with probability at most $\frac{1}{n^{2/3}}$. So with probability at most $5/6$. the number of such comparisons is at most $6c \cdot  n^{2/3} \log(k)$ as there are at most $c n^{4/3}\log(k)$ comparisons in the first round of $A$. In this case, consider items in $Q$. Since $6c \cdot 2 \leq \frac{1}{40}$, at most $\frac{1}{4} n^{2/3}$ are compared to items in $Q$ more than $0.1\log(k)$ times. Therefore $|G| \geq \frac{3}{4} n^{2/3}$. To sum up, we have $\Pr[W_G] \geq 5/6$. 

Now consider $U$. We know that there are at most $c n^{4/3} \log(k)$ comparisons in the second round of $A$. We have $|U| \leq \frac{c n^{4/3} \log(k)}{0.1 \log(k)} = 10c n^{4/3} \leq \frac{1}{32} n^{4/3}$. 

Now we fix $Q$ and $H_1$ such that $W_G$ happens. Consider any sets $S \subseteq \{1,...,k/4\}$. Let $l$ be the largest item in $S$ and let $S' = S \backslash \{l\}$. We also fix $\Pi^{-1}(2r-1), \Pi^{-1}(2r)$ for all $r \in S'$. Now consider the probability such that $Y_l = 1$. For any $\pi$ such that $Y_l =0$, consider $u$ and $v$ such that the followings are satisfied
\begin{itemize}
\item $u \neq v$
\item $u,v \in G$. 
\item $\pi(u) > 2l$, $\pi(v) > 2l$. 
\item $(u,v) \not \in U$. 
\end{itemize}
We have at least $\left(\frac{3}{4} n^{2/3} - k/2\right)^2 - \frac{1}{32} n^{4/3} \geq \frac{1}{32} n^{4/3}$  such $(u,v)$ pairs.

Define a different ranking $\pi^{u,v}$ such that $\pi^{u,v}(u) = 2l-1$, $\pi^{u,v}(v) = 2l$ and the relative positions of other items in $\pi^{u,v}$ is the same as in $\pi$. Notice that with the same $H_1$, switching from $\Pi = \pi$ to $\Pi = \pi^{u,v}$ does not change $Q$ and $\Pi^{-1}(2r-1), \Pi^{-1}(2r), \forall  r\in S'$, but changes $Y_l$ from 0 to 1. We know that each of $u$ and $v$ is compared to items in $Q$ at most $0.1\log(k)$ times in the first round of $A$, we have
\[
\frac{\Pr[H_1 =h_1 | \Pi = \pi]}{\Pr[H_1=h_1|\Pi = \pi^{u,u'}]} \leq \left(\frac{\frac{2}{3}}{\frac{1}{3}} \right) ^{0.2\log(k)} \leq k^{0.2}.
\]
As $\Pr[\Pi = \pi ] = \Pr[\Pi = \pi^{u,v}]$, we have 
\[
\frac{\Pr[\Pi = \pi|H_1 = h_1]}{\Pr[\Pi = \pi^{u,v}|H_1 = h_1]} = \frac{\Pr[H_1 =h_1 | \Pi = \pi]}{\Pr[H_1=h_1|\Pi = \pi^{u,v}]} \leq k^{0.2}.
\]
We know that each $\pi$ corresponds to at least $ \frac{1}{32} n^{4/3}$ such $\pi^{u,v}$ and each $\pi^{u,v}$ corresponds to at most $n^{4/3}$ such $\pi$. Therefore, 
\begin{align*}
&\Pr[Y_u = 0| Q,H_1,\Pi^{-1}(2r-1), \Pi^{-1}(2r),\forall  r\in S']  \\
\leq &\frac{  \frac{1}{32} n^{4/3}}{n^{4/3} \cdot k^{0.2}}  \cdot \Pr[Y_u = 1| Q,H_1, \Pi^{-1}(2r-1), \Pi^{-1}(2r), \forall  r\in S']
\end{align*}
So we have
\[
\Pr[Y_u = 0| Q,H_1,\Pi^{-1}(2r-1), \Pi^{-1}(2r), \forall  r\in S']  \leq 1- \frac{1}{32k^{0.2}}.
\]
And this implies
\[
\Pr[Y_u = 0| Q,H_1, Y_r = 0, \forall r\in S']  \leq 1- \frac{1}{32k^{0.2}}.
\]
For any $S \subseteq \{1,...,k/4\}$, 
\[
\Pr\left[\bigwedge_{l \in S} Y_l = 0|Q,H_1\right] \leq \left(1-\frac{1}{32k^{0.2}}\right)^{|S|}.
\]

Then we have 
\begin{align*}
&\Pr\left[\sum_{l \in [k/4]} Y_r \geq k^{1/4} \right] \\
\geq &\Pr\left[ \sum_{l \in [k/4]} Y_l \geq \frac{k}{4} \cdot \frac{1}{32k^{0.2}} \cdot \frac{1}{2}\right] ~~~\text{(for large enough } k\text{)}\\
\geq &\exp\left(- \frac{k}{4} \cdot \frac{1}{32k^{0.2}} \cdot \frac{1}{2} \cdot \left(\frac{1}{2}\right)^2\right) ~~~\text{(by Theorem \ref{thm:gcb} and Fact \ref{fact:kl})}\\
\geq & 5/6 ~~~\text{(for large enough } k\text{)}\\
\end{align*}

To sum up, we have
\begin{align*}
\Pr[W] &\geq  \Pr\left[\sum_{l \in [k/4]} Y_r \geq k^{1/4} \right] \\
&\geq \sum_{q,h_1: W_G\text{ happens}} \Pr[Q= q,H_1 =h_1] \cdot \Pr\left[\sum_{l \in [k/4]} Y_r \geq k^{1/4} |Q=q,H_1 =h_1\right]\\
&\geq \Pr[W_G] \cdot \frac{5}{6}\\
&\geq \frac{5}{6} \cdot \frac{5}{6} > \frac{2}{3}.
\end{align*}
\end{itemize}

Now we have $\Pr[W] \geq 2/3$. Define $A(H)$ to be the output of the algorithm given history $H$. For a ranking $\pi$ and two labels $(u,v)$, define $\pi^{u,v}$ as the ranking which is the same as $\pi$ except the positions of $u$ and $v$ are swapped. We know that if $\pi(u) \leq k$ and $\pi(v) \leq k$, the sorted top-$k$ of $\pi$ is different from the sorted top-$k$ of $\pi^{u,v}$. For some $\pi$ and $h$, if $W$ happens, define $P(\pi, h)$ as the set of $\sqrt{k}$ pairs of items $(u,v)$ such that $|\pi(u) - \pi(v)| = 1$, $\pi(u) \leq k$ , $\pi(v) < k$ and $(u,v)$ are compared fewer than $0.2 \log(k)$ times in $h$.

Now we have
\begin{align*}
& \Pr[W, A(H) = \text{sorted top-}k \text{ of } \Pi] \\
=& \sum_{\pi, h \text{ s.t. } W \text{ happens}}\Pr[ \Pi = \pi, H = h] \cdot \1_{A(h) = \text{sorted top-}k \text{ of } \pi}  \\ 
\leq &\frac{1}{k^{1/4}} \sum_{\pi, h \text{ s.t. } W \text{ happens}} \sum_{(u,v) \in P(\pi,h)}2^{0.2\log(k)} \cdot \Pr[ \Pi = \pi^{u,v}, H = h] \cdot \1_{A(h) = \text{sorted top-}k \text{ of } \pi} \\
\leq &\frac{2^{0.2 \log(k)}}{k^{1/4}}\sum_{\pi, h \text{ s.t. } W \text{ happens}}\Pr[ \Pi = \pi, H = h] \cdot \1_{A(h) \neq \text{sorted top-}k \text{ of } \pi} \\
<&  \Pr[W, A(H) \neq \text{sorted top-}k \text{ of } \Pi] \\
\end{align*}

Therefore, we have
\begin{align*}
\Pr[W] &= \Pr[W, A(H) = \text{sorted top-}k \text{ of } \Pi] + \Pr[W, A(H) \neq \text{sorted top-}k \text{ of } \Pi] \\
&< 2 \Pr[W, A(H) \neq \text{sorted top-}k \text{ of } \Pi].
\end{align*}

Finally we get
\[
 \Pr[A(H) \neq \text{sorted top-}k \text{ of } \Pi] \geq  \Pr[W, A(H) \neq \text{sorted top-}k \text{ of } \Pi] > \frac{1}{2} \cdot \Pr[W] \geq 1/3.
\]

\end{proof}



\section{Generalized Chernoff Bound}
In the proofs, we use the generalized Chernoff bound of \cite{PS97} as stated below. 

\begin{theorem}[\cite{PS97}]
\label{thm:gcb}
Let $X_1,...,X_n$ be Boolean random variables such that, for some $0\leq \delta \leq 1$, we have that, for every subset $S \subseteq [n]$, $\Pr[\bigwedge_{i \in S} X_i = 1] \leq \delta^{|S|}$, then for any $0\leq \delta \leq \gamma \leq 1$, $\Pr[\sum_{i=1}^n X_i \geq \gamma n] \leq \exp(-n \cdot \D_e(\gamma \| \delta))$. Here $\D_e(\cdot \| \cdot )$ is the Kullback-Leibler divergence defined below.  
\end{theorem}

\begin{definition}[Kullback-Leibler Divergence]
The Kullback-Leibler divergence $\D_e(p\| q) = p \ln\left(\frac{p}{q}\right) + (1-p)\ln\left(\frac{1-p}{1-q}\right)$ for $0\leq p,q \leq 1$. 
\end{definition}

We have the following fact which is used to approximate the Kullback-Leibler divergence.

\begin{fact}
\label{fact:kl}
For $0 \leq p,q\leq 1$, 
\[
\D_e(p\| q) \geq \frac{(p-q)^2}{2\max(p,q)} + \frac{(p-q)^2}{2\max(1-p,1-q)}.
\]
\end{fact}

\end{document}